\documentclass[a4paper,11pt,twoside]{article}
\usepackage[utf8]{inputenc}
\usepackage{amsfonts}
\usepackage{amsmath,amssymb,graphicx,epsfig}
\usepackage{enumerate,amsthm,epstopdf}
\usepackage{bbm}
\usepackage{algpseudocode}
\usepackage{algorithmicx}
\usepackage{algorithm}
\usepackage{listings}
\usepackage[square, numbers]{natbib}
\usepackage{hyperref}
\hypersetup{colorlinks,linkcolor={blue},citecolor={red},urlcolor={blue}}
\usepackage{cleveref}

\usepackage{xcolor}

\usepackage{amsthm} 
\usepackage{amssymb}
\usepackage{mathtools, amsmath}
\usepackage{xcolor}
\usepackage{amsfonts}

\usepackage{epstopdf}
\usepackage[left=1.1in, right=1.1in, top=1.1in, bottom=1.1in]{geometry}
\usepackage{subfig}
\usepackage{booktabs} 
\usepackage{array} 
\usepackage{enumitem}\usepackage{verbatim} 
\usepackage{float}
\usepackage{caption}
\usepackage{bbm} 

\usepackage{physics}
\usepackage{amsmath}
\usepackage{tikz}
\usepackage{mathdots}
\usepackage{yhmath}
\usepackage{color}
\usepackage{array}
\usepackage{multirow}
\usepackage{amssymb}
\usepackage{tabularx}
\usepackage{extarrows}
\usepackage{booktabs}
\usetikzlibrary{fadings}
\usetikzlibrary{patterns}
\usetikzlibrary{shadows.blur}
\usetikzlibrary{shapes}
\usepackage{tcolorbox}
\usepackage{mathtools}
\usepackage{hyperref}
\usepackage{url}
\usepackage[toc,page]{appendix}
\usepackage{color}

\usepackage[utf8]{inputenc}
\usepackage[T1]{fontenc}
\usepackage[english]{babel}

\usepackage[makeroom]{cancel}
\usepackage{authblk}
\usepackage{verbatim}
\usepackage[textwidth=1.8cm, textsize=scriptsize]{todonotes}

\newcommand{\argmax}{\operatornamewithlimits{argmax}}

\newcommand{\de}{^{\delta, \epsilon}}
\newcommand{\ep}{^{\epsilon}}

\def\z{{\mathbf z}}

\def\cP{{\mathcal P}}

\def\bR{{\mathbb R}}

\def\bE{{\mathbb E}}

\def\f0{{\mathbf 0}}

\def\md{{\mathrm{d}}}

\def\R{{\mathbb R}}

\newcommand{\cblue}{\textcolor{blue}}
\newcommand{\zz}{\mathbf{z}}

\def \latentstep {\gamma}
\def \thetastep {\delta}

\newtheorem{thm}{Theorem}

\newtheorem{assumption}{Assumption}

\newtheorem{lem}{Lemma}

\newtheorem{prop}{Proposition}

\theoremstyle{definition}

\newtheorem{rem}{Remark}
\newtheorem{note}{Note}

\newtheorem{example}{Example}

\title{A Multiscale Perspective on Maximum Marginal Likelihood Estimation}
\author[ ]{O. Deniz Akyildiz$ ^{\cblue{\star}}$, Michela Ottobre$ ^{\cblue{\dagger}}$, Iain Souttar$ ^{\cblue{\ddagger}}$}
\affil[$\cblue{\star}$]{Imperial College London}
\affil[$\cblue{\dagger}$]{Heriot-Watt University}
\affil[$\cblue{\ddagger}$]{University of Warwick}
\affil[ ]{{\textcolor{red}{\footnotesize \texttt{deniz.akyildiz@imperial.ac.uk, m.ottobre@hw.ac.uk, iain.souttar@warwick.ac.uk}}}}

\begin{document}
\maketitle
\begin{abstract}
In this paper, we provide a multiscale perspective on the problem of maximum marginal likelihood estimation. We consider and analyse a diffusion-based maximum marginal likelihood estimation scheme using ideas from multiscale dynamics. Our perspective is based on \textit{stochastic averaging}; we make an explicit connection between ideas in applied probability and parameter inference in computational statistics. In particular, we consider a general class of coupled Langevin diffusions for joint inference of latent variables and parameters in statistical models, where the latent variables are sampled from a fast Langevin process (which acts as a sampler), and the parameters are updated using a slow Langevin process (which acts as an optimiser). We show that the resulting system of stochastic differential equations (SDEs) can be viewed as a two-time scale system. To demonstrate the utility  of such a perspective, we show that the  \textit{averaged} parameter dynamics obtained in the limit of scale separation can be used to estimate the optimal parameter,  within the strongly convex setting. We do this by using recent uniform-in-time non-asymptotic  averaging bounds. Finally, we conclude by showing that the  slow-fast algorithm we consider here, termed \textit{Slow-Fast Langevin Algorithm}, performs on par with state-of-the-art methods on a variety of examples. We believe that the stochastic averaging approach we provide in this paper enables us to look at these algorithms from a fresh angle, as well as unlocking the path to develop and analyse new methods using well-established averaging principles.

\end{abstract}

\section{Introduction}\label{sec:intro}
Estimation of parameters in a statistical model with latent variables is one of the most fundamental computational tasks in statistics. This problem arises in many latent variable models described for natural data, such as images, audio, or graphs \cite{bishop2006pattern,smaragdis2006probabilistic,hoff2002latent} or in empirical Bayesian estimation \cite{carlin2000empirical} where the prior can be parameterised. Most notably, the need  arises in complex statistical models and probabilistic unsupervised machine learning \cite{bishop2006pattern}. The latent variable models enable the description of complex statistical relationships via defined hidden variables as well as also enable the ``discovery'' of latent representations of data \cite{whiteley2022discovering}. Therefore, it is crucial to develop principled and efficient methods for parameter estimation in latent variable models with precise theoretical guarantees.

Statistically speaking, the estimation of parameters under the presence of latent variables can be generally named as ``maximum marginal likelihood estimation'' (MMLE) problem \cite{dempster1977maximum}, as the likelihood of data in this case is defined as an integral over the unobserved variables. In this setting, we define our parameterised statistical model as $p_\theta(x, y)$ which involves three key quantities: A parameter $\theta \in \bR^{d_\theta}$, latent variables $x \in \bR^{d_x}$, and \textit{observed} data $y \in \bR^{d_y}$. We view the statistical model $p_{\theta}(x,y)$ as a function of $x$ and $\theta$ with fixed (observed) $y$, i.e., $p_\theta(x, y): \bR^{d_\theta} \times \bR^{d_x} \to \bR$. Given observed (and fixed) data $y$, the marginal likelihood is defined as an integral $p_\theta(y) = \int_{\R^{d_x}} p_\theta(x,y) \mathrm{d}x$,
where $p_\theta(y): \mathbb{R}^{d_\theta} \to \bR$ is seen as a function of $\theta$, as $y$ is fixed. Then the problem of MMLE can be simply defined as the maximisation  with respect to $\theta$ of this function (referred to as the objective function) or,  more generally, of $\log p_\theta(y)$. The purpose is to find \textit{marginal} maximum likelihood estimates $\theta^{\star}$ that best explain the observed data $y$. Unlike the case in maximum likelihood estimation (MLE), the structure of this problem excludes the use of classical optimisation techniques as the objective function is an integral, which can in general not be calculated exactly, so that an explicit expression for $p_{\theta}(y)$ is not available. 

One of the most common approaches for solving this kind of problem is the celebrated Expectation Maximisation (EM) algorithm \cite{dempster1977maximum} (which we recall in more detail in Section \ref{sec: expectation- maximization algorithm}). The EM algorithm is an iterative method that alternates between two steps: The E-step, where the expectation of the log-likelihood is computed with respect to the latent variables, and the M-step, where the parameters are updated using the expectation computed in the E-step. The EM algorithm is a general method that can be applied to a wide range of latent variable models, and it is known to converge to a local maximum of the marginal likelihood \cite{wu1983convergence}. Furthermore, within the setting of general statistical models, the EM algorithm contains two generally intractable steps, namely the computation of the expectation in the E-step and the maximisation in the M-step. Traditionally, the implementations focused on models where these steps are exact, but in the last few decades, a plethora of methods have been proposed to replace these steps with approximate procedures. For example, instead of computing an exact expectation in the E-step, one can approximate this step with perfect Monte Carlo \cite{wei1990monte,sherman1999conditions} if this simulation task is feasible. In the same vein, the M-step can be performed approximately, using numerical techniques such as gradient descent \cite{meng1993maximum,liu1994ecme,lange1995gradient}. Based on these ideas, in the last two decades, significant research activity was focused on developing approximate EM algorithms, see, e.g., Monte Carlo EM (MCEM) \cite{wei1990monte}, stochastic EM (SEM) using stochastic approximation \cite{celeux1985sem} and several variants, e.g., \cite{celeux1992stochastic,chan1995monte,booth1999maximizing,cappe1999simulation,diebolt1995stochastic}. These methods often replace the E-step with a perfect Monte Carlo estimate and approximate the M-step with stochastic approximation techniques. Naturally, in complex statistical models, E-step would be defined with respect to the posterior density of latent variables, which would typically be impractical. This impracticality motivated the use of Markov chain Monte Carlo (MCMC) methods to approximate the E-step, see, e.g., \cite{delyon1999convergence,fort2011convergence,fort2003convergence,caffo2005ascent,atchade2017perturbed}.

In recent years, motivated by diffusion (or SDE) based MCMC samplers, such as the unadjusted Langevin algorithm (ULA) \cite{roberts1996exponential,durmus2019high,dalalyan2017further,dalalyan2017theoretical,dalalyan2019user}, a number of methods have been proposed for MMLE problem that are approximating the E-step of the EM with an unadjusted Langevin chain. In this context, \cite{de2021efficient} studied an algorithm, abbreviated as SOUL, which consists of an E-step based on an approximation provided by a ULA chain and an M-step based on gradient descent. The work \cite{de2021efficient} showed the convergence of the algorithm (in discrete-time) under some strict conditions, and furthermore provided empirical evidence of the performance of this method. {Inspired by SDE-based approaches such as SOUL, there has been also algorithms developed based on \textit{interacting} particle systems - which instead of a single ULA chain for the E-step, use a system of interacting $N$ particles run in parallel with parameter updates, see, e.g., \cite{kuntz2023particle, akyildiz2023interacting}. These methods, termed particle gradient descent (PGD) \cite{kuntz2023particle} and interacting particle Langevin algorithm (IPLA) \cite{akyildiz2023interacting}, are related to the SOUL (and other EM methods above) as they maximise the marginal likelihood. They are also related to our approach as they build SDEs defined in the joint space of parameters and the latent space.}

We should emphasize that, {from our perspective}, the SOUL algorithm, as much as the work of the present paper, can be seen as a reincarnation, in a statistical context, of algorithms such as the {\em Heterogeneous Multiscale Method} (HMM) \cite[Chapter 9]{weinan2011principles}, {\em equation free computation} (EFC) \cite{kevrekidis2009equation} or {\em Temperature Accelerated Molecular Dynamics} (TAMD) \cite{stoltz2018longtime}, which have been well known for a long time in molecular dynamics. In this paper we will mostly make comparisons with SOUL (and other recent algorithms) rather than e.g. HMM, but this is only because the context and language used in the literature on HMM, EFC or TAMD are typically more adapted to molecular dynamics, while SOUL is phrased in statistical terms; nonetheless  all these algorithms, HMM, EFC, TAMD, SOUL  (to name just a few and with no claim to completeness of references) significantly overlap,  so we could have equally chosen HMM for comparison.

One of the main purposes of this paper is to show that a stochastic averaging perspective can help reach a unified view on the algorithms we mentioned so far (and probably on others as well). Beyond being conceptually helpful, this observation can help create frameworks for analysis and stimulate ideas for new algorithms tackling the MMLE problem.\\

\noindent\textbf{Contributions.} In this paper, we consider the MMLE problem and observe that it is possible to design a multiscale SDE system which, in appropriate  regimes, converges to the sought after optimiser, $\theta^{\star}$, of the marginal likelihood $p_{\theta}(y)$. More precisely, the multiscale system we consider is a so-called slow-fast system of SDEs, where the slow variables serve to update the parameter process while the fast variables correspond to the latent variables of our statistical model.   We leverage recent uniform-in-time averaging results (see, e.g., \cite{crisan2022poisson}) to provide a novel analysis of the resulting two-time scale system. In particular, we show that, in the statistical setting, the parameter process of the two-time scale system can be shown to stay close to an ideal noisy gradient ascent scheme which includes the gradient of the log marginal likelihood. That is, the slow parameter process will converge, in the small diffusion limit, to the desired optimiser (assuming such an optimiser is unique). We then argue that this two-time scale system can be viewed as a continuous-time limit of various expectation-maximization type algorithms, such as SOUL.  This enables us to provide a novel perspective on the convergence of these algorithms. More explicitly, in this paper we consider Euler discretizations of our two-scale continuous time dynamics - and we study the discretization error as well - and suggest that SOUL and, more in general, EM, can be seen as different discretizations of this underlying  slow-fast system of SDEs.

As a byproduct of our approach, we unify a number of known algorithms from computational statistics and multiscale methods such as SOUL \cite{de2021efficient}, IPLA \cite{akyildiz2023interacting}, PGD \cite{kuntz2023particle}, and HMM \cite{weinan2011principles} EFC \cite{kevrekidis2009equation}, and (TAMD) \cite{stoltz2018longtime} within a common framework, observing that they could all be thought of as different discretizations of an appropriate multiscale dynamics. This observation unlocks a number of research directions in putting forward a coherent theoretical framework for these methods.

The paper is  organised as follows. In Section~\ref{sec:background}, we clarify the problem at hand and provide background material about the EM algorithm and about Stochastic Averaging. In Section \ref{sec:stoch_av_approach to MMLE} we explain how two-scale systems of SDEs can be used to address the MMLE problem. In Section \ref{sec:main results} we state our main results, namely we present {\em non-asymptotic} error bounds to prove the convergence of our continuous-time slow-fast system to the desired optimiser and  to estimate the discretization error when our continuous time dynamics is discretised via Euler-Maruyama scheme. The analysis is carried out in the case of unique optimiser (but we make some remarks on the more general case of multiple optimisers).  Section \ref{sec:relevant_work} contains more precise comparison to related results in the literature while in Section \ref{sec:comparison_SOUL} we spell out the relationship between our (continuous time) multiscale system and the (discrete time) SOUL algorithm.  Finally, Section \ref{sec:sec 4} gathers precise assumptions under which our main results hold and  proofs of such results; Section \ref{sec:numex} provides numerical examples on the performance of our algorithm (i.e. on the Euler-Maruyama discretization of the multiscale SDE we discuss), in comparison to SOUL or to particle-based algorithms such as PGD. All proofs can be found in \cite{suppMat}.

\section{Background}\label{sec:background}
\subsection{Problem definition}
Let $p_\theta(x,y)$ be a generic statistical model. We assume that the data $y$ is \textit{fixed} from now on and will treat $p_\theta(x, y)$ as a function of $(\theta, x)$ only. Given a statistical model $p_\theta(x, y)$, the usual tasks of statistical inference are (i) to estimate the parameters $\theta$ that best explain the data, (ii) then subsequently sample from the conditional (posterior) distribution $p_{\theta^\star}(x|y)$ to conduct inference on the latent variables. The first task is known as the maximum marginal likelihood estimation (MMLE) problem, and the second task is known as the sampling problem.

The MMLE problem is defined as the following maximisation problem
\begin{align}\label{eq:mmle}
    \theta^\star \in \arg\max_{\theta \in \bR^{d_\theta}} \log p_\theta(y) \, ,
\end{align}
where $p_\theta(y) := \int p_\theta(x,y) \mathrm{d}x$ is the marginal likelihood and is treated as a function of $\theta$, i.e., $p_\theta(y): \bR^{d_\theta} \to \bR$. We next define the negative joint log-likelihood
\begin{align}\label{potential}
    U(\theta, x) = -\log p_\theta(x, y). \footnotemark 
\end{align}
\footnotetext{Of course we should retain the dependence on the data $y$ also in the notation for the potential $U$. We refrain from doing so to follow tradition in the field. }
Given this definition, we can rewrite $p_\theta(y)$ as
\begin{align}
    p_\theta(y) := \int_{\mathbb{R}^{d_x}} p_\theta(x,y) \mathrm{d}x = \int_{\mathbb{R}^{d_x}} e^{-U(\theta,x)} \mathrm{d}x.
\end{align}
Maximisation of $p_\theta(y)$ is a central statistical problem. One of the most used methods to this end is the celebrated expectation maximisation algorithm which is introduced next \cite{dempster1977maximum}.

\subsection{The Expectation-Maximisation algorithm}\label{sec: expectation- maximization algorithm}
{While our approach can be motivated solely using the MMLE problem, we introduce the EM algorithm below to make the connection between the framework introduced in this paper and practical implementations of the EM algorithm concrete.}

The EM procedure \cite{dempster1977maximum} can be succinctly described as follows. Given an estimate of the parameter at time $k-1$, denoted as $\theta_{k-1}$, the idea is to produce the next iterate $\theta_k$ to maximise $p_\theta(y)$ in a two-stage iterative algorithm. These two stages are called Expectation (E) and Maximisation (M) steps, respectively. First, the E-step is performed to compute an expectation (E-step)
\begin{align}\label{eq:theEM-E-Step}
Q(\theta, \theta_{k-1}) = \mathbb{E}_{p_{\theta_{k-1}}(x|y)}[\log p_\theta(x,y)] = - \int U(\theta, x) p_{\theta_{k-1}}(x|y) \mathrm{d}x,
\end{align}
which is followed by the maximisation step (M-step)
\begin{align}\label{eq:theEM-M-Step}
\theta_k \in \arg \max_{\theta \in \R^{d_{\theta}}} Q(\theta, \theta_{k-1}). \quad \quad \quad \quad
\end{align}
This strategy works, as one can easily show that $\log p_{\theta_k}(y) \geq \log p_{\theta_{k-1}}(y)$ (see Appendix~\ref{app:EM}) and the method hence climbs to a local maximum. However, these two steps are generally intractable. In particular, the E-step in \eqref{eq:theEM-E-Step} requires integration w.r.t. the posterior distribution $p_{\theta}(x|y)$ which is unavailable in general. This step is therefore generally approximated with sampling algorithms (and in the most general case, Markov chain Monte Carlo methods). In turn, the resulting $Q$ function is an approximation, which is maximised in \eqref{eq:theEM-M-Step}. Similarly, however, this maximisation is also intractable in most general case and typically approximated using numerical optimisation methods, such as gradient descent \cite{meng1993maximum, liu1994ecme}. We will show that the combination of the two steps can be interpreted as a form of stochastic averaging, which we come to review in the next section.

\subsection{Background on stochastic averaging}\label{sec:sa_background}
Let us briefly recap on some basics of stochastic averaging, in  the SDE setting which is relevant to this paper. For a more general introduction we refer to the excellent textbooks \cite{weinan2011principles, pavliotis2006introduction}. To fix ideas, consider the following  system of SDEs
\begin{align}
\mathrm{d} \theta_t &= b(\theta_t, X_t) \mathrm{d}t + \sqrt{2\sigma} \mathrm{d}W_t^0 \label{eqn:slowintro}\\
\mathrm{d} X_t &= \frac{1}{\epsilon}  a(\theta_t, X_t) \mathrm{d}t + \sqrt{\frac{2}{\epsilon}} \mathrm{d}W_t^1 \label{eqn:fastintro} 
\end{align}
where $0<\epsilon \ll 1$, the  $W_t^i$'s are independent, $d_{\theta}$ and $d_{x}$  - dimensional standard Brownian motions, respectively, while $b: \mathbb R^{d_{\theta}} \times \mathbb R^{d_{x}} \rightarrow \mathbb R^{d_{\theta}}$, $a: \mathbb R^{d_{\theta}} \times \mathbb R^{d_{x}} \rightarrow \mathbb R^{d_{x}}$ are  coefficients satisfying appropriate assumptions  (which we make precise in subsequent sections) and $\sigma>0$.

 Due to the time-scaling properties of Brownian motion,  as $\epsilon$ tends to zero the process $X_t$ in \eqref{eqn:fastintro} moves faster and faster (for $t$ fixed), so that the parameter $\epsilon$ defines a scale separation between the {\em slow process} (SP) $\theta_t$ and the {\em fast process} $X_t$. For this reason systems of the form \eqref{eqn:slowintro}-\eqref{eqn:fastintro} are often referred to as {\em slow-fast } systems of SDEs. In this context one is usually  interested in understanding what is the behaviour of the process $\theta_t$, in the (singular) limit $\epsilon \rightarrow 0$. The study of this limit goes under the name of {\em (stochastic) averaging}.    Since $X_t$ evolves much faster than $\theta_t$, intuitively, the former process will have ‘reached equilibrium’ while the latter has remained substantially unchanged. Hence, one expects that, in the limit $\epsilon \rightarrow 0$, the process $\theta_t$ will only be influenced by the equilibrium state of the process $X_t$. Since these two processes are coupled, one way to understand how to obtain the averaging limit is as follows:   we fix the value of the SP, say to $\theta_t=\theta$, and consider an intermediate process, the so-called 'Frozen Process' (FrP), namely the process
\begin{equation}\label{eqn:frozenintro}
\md X_t^{\theta} = a( \theta,X^{\theta}_t)  \md t + \sqrt{2} \md W_t^1
\end{equation}
which can be seen as a one-parameter family of processes, parametrized by $\theta$. It  is important to notice at this point that \eqref{eqn:frozenintro} is obtained from \eqref{eqn:fastintro} by setting $\epsilon=1$ and freezing the value of $\theta_t$. To understand what comes next note also that, once in \eqref{eqn:fastintro} we freeze the value of $\theta_t$, again because of the time-scaling of Brownian motion, fixing $t$ and letting $\epsilon$ to zero is the same as fixing $\epsilon$ and letting $t$ to infinity. 
With this in mind,  assuming  that the FrP  is ergodic, i.e. that $X^{\theta}_t$  admits, for every $\theta$ fixed, a unique invariant measure, $\mu^{\theta}$,  (and moreover that the process converges to such an invariant measure irrespective of its initial state as $t \rightarrow \infty$) one can show that, as $\epsilon$ tends to zero, the SP $\theta_t$, solution of \eqref{eqn:slowintro},   converges, at least weakly,  and at least over  finite time-horizons $[0,T]$,  to the so called {\em averaged dynamics}, $\bar{\theta}_t$, namely to the solution of the following SDE

\begin{align}\label{eq:general_averaged_sde}
\md\bar \theta_t = \bar b (\bar \theta_t) \md t + \sqrt{2\sigma} \md W_t\, , 
\end{align}
where the drift $\bar b$  is obtained from the one of the SP in \eqref{eqn:slowintro} by averaging the fast variable with respect to the invariant measure, i.e.
\begin{equation}\label{barb}
\bar b (\theta) = \int_{\mathbb R^{d_x}  }b(\theta,x) \mu^{\theta}(\md x) \,.
\end{equation}

To summarize, via the averaging procedure we have just described, one can approximate the slow-fast process \eqref{eqn:slowintro}-\eqref{eqn:fastintro} - more precisely, the process \eqref{eqn:slowintro} -  via the averaged dynamics \eqref{eq:general_averaged_sde}. The dynamics \eqref{eq:general_averaged_sde} can be thought of as a reduced description of \eqref{eqn:slowintro}-\eqref{eqn:fastintro} (it is lower dimensional and it does not retain the detail of the fast process, as it only sees the behaviour in equilibrium of $X_t$), and this point of view is often adopted in modelling contexts, see \cite{weinan2011principles, pavliotis2006introduction}. In computational statistics it is however more advantageous to adopt the reverse point of view and use the dynamics \eqref{eqn:slowintro}-\eqref{eqn:fastintro} as a way of simulating \eqref{eq:general_averaged_sde}. This is the perspective we adopt in this paper.  In particular in this paper we will want to simulate SDEs of the form \eqref{eq:general_averaged_sde}, motivated by the fact that, when the coefficients  $a$ and $b$ are chosen appropriately, the long-time behaviour of \eqref{eq:general_averaged_sde}
is described by our marginal of interest $p_{\theta}(y)$. This is explained in the next section but we anticipate it here to motivate the following discussion. 

The convergence of \eqref{eqn:slowintro} to the solution of \eqref{eq:general_averaged_sde} has been proved in various contexts, using a range of methods and at many levels of generality, see e.g.  \cite{pavliotis2008multiscale,GlynnMeyn, ethierKurtz}  or also \cite{pardoux2001, pardoux2003,LiuRockner, rockner2020diffusion}, which treat the setup in which the processes' state space is non-compact. See also Section \ref{sec:relevant_work} for further literature review. 

However, before the works \cite{crisan2022poisson, barre2021fast} convergence (whether weak
or strong) of the slow-fast system to the averaged dynamics, had  only been proved to take place over finite time horizons. That is,
broadly speaking, one typically establishes results of the type
\begin{equation}\label{uitdefinition}
   \left\vert\mathbb E f(\theta_t) - \mathbb E f(\bar{\theta}_t) \right \vert \leq C \epsilon^{\gamma} 
\end{equation}
for every $f$ in a suitable class of functions and some $\gamma>0$,   with $C = C(t)$ a constant dependent on time (and on the function $f$ as well). Because of the use of Gronwall’s inequality (or similar arguments), this constant is typically an increasing function of time $t$. From an algorithmic perspective this creates a clear issue as in principle one would have to choose a smaller $\epsilon$ for  longer computation time, if a certain threshold accuracy is desired (more comments on this and on other matters related to uniform in time bounds in Section \ref{sec:main results}).  More broadly,  while one can typically only prove that the
averaged dynamics is a good approximation of the original slow-fast system for finite-time
windows, with estimates that deteriorate in time, the slow-fast system  is often used
in practice as an approximation of the long-time behaviour of the averaged dynamics -- which is indeed what we do in this paper. This is clearly not justified if $C(t)$ increases with time and indeed  the need in applications
for multiscale results which hold uniformly in time,   has
been explicitly advocated, see e.g. \cite{pavliotis2022derivative}.  
The work \cite{crisan2022poisson}  filled  this theoretical gap and identified rather general
assumptions on the coefficients of the SDEs \eqref{eqn:slowintro}-\eqref{eqn:fastintro}, under which convergence, as $\epsilon$ tends to zero,   of the
slow-fast system to the limiting dynamics is actually uniform in time (UiT) - by which we mean, for the sake of clarity,  that \eqref{uitdefinition} holds for some constant $C$ independent of $t$. In this work we leverage on this result to provide rigorous foundations for multiscale-based MMLE methods. The results in \cite{crisan2022poisson} were made possible by a novel stability analysis, based on the concept of Strong Exponential Stability (SES). SES is a flexible tool, which can be used not only for approximations produced via multiscale methods but in a variety of contexts, including approximations produced via numerical methods \cite{angeli2023uniform, crisan2021uniform}, via particle methods \cite{barre2021fast} and  for infinite dimensional dynamics 
\cite{dobson2023infinite}. 

Finally, we point out that estimate \eqref{uitdefinition} is non-asymptotic. Uniform in time bounds which are asymptotic in time have also been studied. This is discussed in Section \ref{sec:relevant_work}, see \eqref{eqn:gabriel} for comparison.

\section{Stochastic averaging approach to MMLE}\label{sec:stoch_av_approach to MMLE}
In this section, we formalise our setup. Recall that our aim is to maximise $\log p_\theta(y)$, i.e., to solve the problem \eqref{eq:mmle} to find the maximisers of $\log p_\theta(y)$. To this end, it is natural to consider the following diffusion process 
\begin{align}\label{eq:averaged_theta_bar}
    \mathrm{d}\bar{\theta}_t = \nabla_\theta \log p_{\bar{\theta}_t}(y) \mathrm{d}t + \sqrt{\frac{2}{\beta}} \mathrm{d}W_t^0 \,.
\end{align}
The stationary measure of this diffusion, namely the measure 
\begin{equation}\label{eqn:rho-beta}
\bar{\rho}_\beta(\theta) \propto e^{\beta \log p_\theta(y)}
\end{equation}
will concentrate on the maximisers of $\log p_\theta(y)$, as $\beta \rightarrow \infty$ \cite{hwang1980laplace}. 
 Hence this diffusion can be seen as an \textit{optimiser} of $\log p_\theta(y)$ as  $t\rightarrow \infty$ and $\beta \to \infty$. This fact has been well-studied;  in particular, even nonasymptotic rates w.r.t. $\beta$ can be obtained, depending on the assumptions on $\log p_\theta(y)$, i.e., $\mathcal{O}(1/\sqrt{\beta})$ for the strongly convex case \cite{akyildiz2023interacting} or $\mathcal{O}(\log \beta / \beta)$ for the dissipative case \cite{pmlr-v65-raginsky17a}.

However, simulating \eqref{eq:averaged_theta_bar} is impractical  as we do not have access to $\nabla_\theta \log p_\theta(y)$. Conveniently, observing that
\begin{align*}
\nabla_\theta \log p_\theta(y) &= \frac{\nabla_{\theta} p_\theta(y)}{p_{\theta}(y)} = 
\frac{\nabla_{\theta} \int_{\mathbb R^{d_x}}p_{\theta}(x,y) \, \md x}{p_{\theta}(y)} \\
&=
\frac{\int \left[\nabla_{\theta} \log p_\theta(x,y)\right] p_{\theta}(x,y) \md x }{p_\theta(y)} = -\int \nabla_\theta U(\theta, x) p_\theta(x | y) \mathrm{d} x,
\end{align*}
one can write
\begin{align}\label{eq:theta_bar_sde}
    \mathrm{d} \bar{\theta}_t &=-\left[\int \nabla_\theta U({\bar{\theta}_t}, x) p_{\bar{\theta}_t}(x|y) \mathrm{d}x \right]\mathrm{d}t + \sqrt{\frac{2}{\beta}} \mathrm{d}W^0_t.
\end{align}
This rewriting then allows us to relate the above dynamics to \eqref{eq:general_averaged_sde}, i.e. to interpret \eqref{eq:theta_bar_sde} as the ``averaged'' process of an appropriate two-scale dynamics, namely of the following slow-fast system of SDEs, 
\begin{align}
    \mathrm{d}\theta_t &= -\nabla_\theta U(\theta_t, X_t) \mathrm{d}t + \sqrt{\frac{2}{\beta}} \mathrm{d}W^0_t,\label{eq:theta_update}\\
    \mathrm{d}X_t &= -\frac{1}{\epsilon}\nabla_x U(\theta_t, X_t) \mathrm{d}t + \sqrt{\frac{2}{\epsilon}} \mathrm{d}W^1_t\label{latent_update}
\end{align}
 where $0<\epsilon \ll 1$. 
 Indeed, note that if we fix the value of $\theta_t$ in \eqref{latent_update} then we see that 
 the `Frozen Process' associated to \eqref{latent_update} is given by (cfr \eqref{eqn:frozenintro})
\begin{equation}\label{eqn:frozen}
        \mathrm{d}X^\theta_t = -\nabla_x U(\theta, X^\theta_t) \mathrm{d}t + \sqrt{2} \mathrm{d}W^1_t
\end{equation}
and, for each $\theta \in \R^{d_{\theta}}$ fixed, the (normalised) invariant measure of the above process is given by ${e^{-U(\theta, x)}}/{\int_{\R^{d_x}} e^{-U(\theta, x) }dx}$. Using \eqref{potential}, it is clear that, from the point of view of our statistical model $p_{\theta}(x,y)$, the above measure is precisely $p_{\theta}(x\vert y)$. In other words, we can use the system \eqref{eq:theta_update}-\eqref{latent_update} as an approximation of \eqref{eq:theta_bar_sde}.  The advantage is that while \eqref{eq:theta_bar_sde} cannot be simulated directly, the system \eqref{eq:theta_update}-\eqref{latent_update} is, in principle, more amenable to simulation.  Hence, running the dynamics  \eqref{eq:theta_update}-\eqref{latent_update} for $\beta$ large and for long time $t$, will give access to the desired optimiser $\theta^{\star}$. 

The backbone of the approach in this paper is  the fact  that we can provide UiT convergence of \eqref{eq:theta_update} to \eqref{eq:theta_bar_sde} as $\epsilon  \rightarrow 0$. This is done in Theorem \ref{thm:main_thm}, stated in Section \ref{sec:main results}. Such a result guarantees  that the long time behaviour of \eqref{eq:theta_update} is "close" to the long time behaviour of \eqref{eq:theta_bar_sde}, as $\epsilon \rightarrow 0$, with estimates that  are independent of $\beta$.  These facts are crucial, given that we are interested in the long time behaviour of \eqref{eq:theta_bar_sde}, for $\beta$ large.  

In order to obtain a practical algorithm out of \eqref{eq:theta_update}--\eqref{latent_update}, we still need to discretise such a system. To this end,  in our numerical experiments  we will  use the standard Euler-Maruyama discretisation, namely 
\begin{align}
\theta_{k+1}^{\delta} &= \theta_{k}^{\delta} - \delta \nabla_{\theta} U(\theta_k^{\delta}, X_k^{\delta}) + \sqrt{\frac{2\delta}{\beta}}  \xi_k^1 \label{eq:euler-Maruyama-1} \\
X_{k+1}^{\delta} &= X_{k}^{\delta} -\frac{\delta}{\epsilon} \nabla_x U(\theta_k^{\delta}, X_k^{\delta}) + \sqrt{\frac{2\delta}{\epsilon}} \xi_k^2, \label{eq:euler-Maruyama-2}
\end{align}
where $\xi_k^1, \xi_k^2$ are i.i.d multivariate Normal random variables\footnote{Clearly both $\theta^{\delta}_k$ and $X_k^{\delta}$ depend on $\epsilon$ and $\beta$ as well, so a better notation for these quantities would be $\theta^{\delta, \epsilon, \beta}_k X_k^{\delta, \epsilon, \beta}$. We refrain from using such notation to avoid cluttering.}.  We emphasize that Euler-Maruyama is the simplest choice, and  it comes with drawbacks, discussed in the next section.  One could certainly choose other discretization schemes. Here we fix this choice purely for simplicity and to allow comparison with the other approaches to MMLE, such as SOUL and PGD. We refer to the algorithm given in \eqref{eq:euler-Maruyama-1}--\eqref{eq:euler-Maruyama-2} as \textit{Slow-Fast Langevin Algorithm} (SFLA). 
The discretization error produced when approximating the dynamics \eqref{eq:theta_update}-\eqref{latent_update} with Euler-Maruyama is studied in Proposition \ref{prop:numerical_error}, see Section \ref{sec:main results}.

\begin{algorithm}[t]
\caption{Slow-Fast Langevin Algorithm (SFLA)}
\label{alg:ais}
\begin{algorithmic}[1]
\State Choose $\epsilon, \delta, \beta, X_0, \theta_0, K$.
\For{$k = 1, \ldots K$}
\State Sample  $\xi_k^1 \sim \mathcal N(0,Id_{d_{\theta}})$ and then calculate $\theta_{k+1}^{\delta}$ as follows: 
\begin{align*}
\theta_{k+1}^{\delta} &= \theta_{k}^{\delta} - \delta \nabla_{\theta} U(\theta_k^{\delta}, X_k^{\delta}) + \sqrt{\frac{2\delta}{\beta}}  \xi_k^1,
\end{align*}
\State Sample $\xi_k^2 \sim \mathcal N(0,Id_{d_{x}})$ and then calculate $X_{k+1}\de$ as follows: 
\begin{align*}
X_{k+1}^{\delta} &= X_{k}^{\delta} -\frac{\delta}{\epsilon} \nabla_x U(\theta_k^{\delta}, X_k^{\delta}) + \sqrt{\frac{2\delta}{\epsilon}} \xi_k^2,
\end{align*}
\EndFor
\State Report $\theta_K$.
\end{algorithmic}
\end{algorithm}

\begin{note}
     We would like to note at this point the practicality of the continuous-time formulation in \eqref{eq:theta_update}--\eqref{latent_update} and their discretization \eqref{eq:euler-Maruyama-1}--\eqref{eq:euler-Maruyama-2}. This system models an implementation of the expectation-maximisation scheme, which uses Markov chain Monte Carlo (MCMC) dynamics in its E-step and uses a noisy gradient ascent scheme to perform the M-step. The introduction of $\epsilon$ further models a real world choice: the choice of different step sizes. It is common in simulations where practitioners choose different step sizes within this setting \cite{de2021efficient} which, in effect, results in a two time-scale system. Here our framework directly accommodates this. We will expand more on this point, namely on relating the two time scale system as a framework for realising the EM algorithm in Section \ref{sec:comparison_SOUL}. There we will use the specific setting of SOUL.  
\end{note}

\section{Main results}\label{sec:main results}
In this section we first state the main results of the paper and then comment on relation to literature. 

The main results of this paper  are Theorem \ref{thm:main_thm}, Proposition  \ref{prop:numerical_error}  and Theorem \ref{thm:combined result}; the latter combines  the estimates of Theorem \ref{thm:main_thm} and Proposition  \ref{prop:numerical_error}. 

Theorem \ref{thm:main_thm}  shows   convergence of the dynamics $\theta_t$, solution of \eqref{eq:theta_update}, to the optimiser $\theta^{\star}$, as $\epsilon\rightarrow 0$ and  $\beta, t\rightarrow \infty$. The bounds contained in this theorem are non-asymptotic.  The section is organised as follows: we first state Theorem \ref{thm:main_thm}  and  make comments on how it is obtained. We then state Proposition \ref{prop:numerical_error}, which addresses the study of the discretization error produced when approximating the dynamics \eqref{eq:theta_update} with Euler-Maruyama. The two results are then combined in Theorem \ref{thm:combined result},  which states the convergence of $\theta_k^{\delta}$, the Euler approximation of $\theta_t$, to $\theta^{\star}$. In this section we state results and make various comments on them. 
Proofs and full statements of assumptions  are gathered in Section \ref{sec:sec 4}.

Before stating Theorem \ref{thm:main_thm} let  us point out that the solution $\theta_t$ of \eqref{eq:theta_update} depends on $\epsilon, \beta$ and on the initial data $\theta_0, X_0$ of system \eqref{eq:theta_update}-\eqref{latent_update}, so we should denote it by $\theta_t^{\epsilon, \beta, \theta_0, X_0}$. We refrain from doing so to avoid notation overload but it is important to keep this in mind while reading Theorem \ref{thm:main_thm} below. 
In the following, and throughout, we use $\mathcal C^2_b$ to denote the space of bounded and twice continuously differentiable functions (and, where required, denote the corresponding Lipschitz constant as $C_f$).
\begin{thm}\label{thm:main_thm}
Let $\theta_t$ be as in \eqref{eq:theta_update} and $\theta^{\star}$ be  as in \eqref{eq:mmle}. Let Assumption \ref{ass:growthcoeffs}, Assumption \ref{assmp:convexity} ,  Assumption \ref{ass:strongconvexfast} and Assumption \ref{ass:avgderest} (all of which stated in Section \ref{subsec:assumptions} below) hold. Then, for any $f:\R^{d_{\theta}} \rightarrow \R$, $f \in \mathcal C^2_b$, the following holds: 
\begin{equation}\label{main estimate}
    \left \vert \mathbb E f(\theta_t) - f(\theta^{\star})\right \vert^2\leq \epsilon C (\|\nabla f\|^2_{\infty}+ \|\nabla^2f\|_{\infty}^2) +C_f^2\left(e^{-2\mu t} |{\theta}_0 - {\theta}^\star|^2 +  \frac{2 d_\theta}{\beta} (1 - e^{-2\mu t})\right),
\end{equation}
with $C = C_0(1+|X_0|^2+|\theta_0|^2)$ where $C_0$ is a constant independent of $\beta, \theta_0$ and $X_0$;  $\mu$ is as in Assumption \ref{assmp:convexity}, $|\cdot|$ is the Euclidean norm and $\|\cdot \|_{\infty}$ is the supremum norm. 
\end{thm}
While the precise assumptions under which the above theorem holds are detailed in the next section, we anticipate that the two main requirements are for $U$ to be Lipschitz and convex. These assumptions can be relaxed, but this is not the main purpose of this paper.

Let us now make a few comments on the above result. The estimate \eqref{main estimate} (almost) implies convergence of the law of $\theta_t$ to the dirac measure centered in $\theta^{\star}$. As a word of caution, this is not exactly weak convergence, as we prove \eqref{main estimate} to hold for every $f$ in  
$\mathcal C^2_b$, not for every bounded measurable function, which is what is required to prove weak convergence. This is substantially to streamline proofs. \footnote{To prove the same result for every $f$ continuous and bounded one needs to use appropriate smoothing estimates; we made more comments on this point in \cite{crisan2022poisson, barre2021fast}, see e.g. \cite[Note 2.2]{barre2021fast}, so we don't repeat them here. } Nonetheless the intuitive meaning remains analogous.  Indeed,  suppose for a moment we were allowed to take $f(\theta)= \mathbf 1_{A}(\theta)$, for some set $A \subseteq  \R^{d_{\theta}}$; then with this choice of $f$, the left hand side (LHS) of  \eqref{main estimate} is equal to $\left\vert\mathbb P (\theta_t \in A) -  \chi_{\theta^{\star}} (A)\right\vert^2$ where $\chi_{\theta^{\star}}(A)$ is equal to one if $\theta^{\star} \in A$ and to zero otherwise. 
To explain where the three contributions on the RHS come from, let us observe that the  bound \eqref{main estimate} is obtained in two steps, corresponding to estimating the two addends on the RHS of the triangular inequality below: 
\begin{align} \label{triangular inequ}
 \left \vert \mathbb E f(\theta_t) - f(\theta^{\star})\right \vert^2 &\leq    \left \vert \mathbb E f(\theta_t) - \mathbb E f(\bar{\theta}_t)\right \vert^2 +  \left \vert \mathbb E f(\bar{\theta}_t) - f(\theta^{\star})\right \vert^2 \, 
\end{align}
where $\bar{\theta_t}$ is as in \eqref{eq:theta_bar_sde}. 
That is, we first estimate the `distance' between the process $\theta_t$ and the averaged evolution ${\bar \theta_t}$. Then, since for $\beta$ and $t$ large $\bar \theta_t$ converges to $\theta^{\star}$, we  estimate the `distance' between $\bar \theta_t$  and  $\theta^{\star}$. 
The estimate on the first addend on the RHS of \eqref{triangular inequ} results in the first term on the RHS of \eqref{main estimate}, showing that  $\theta_t$ tends to  $\bar \theta_t$ as $\epsilon \rightarrow 0$, and the convergence is UiT (more comments on this below). This first estimate is stated and proved in Section \ref{subsec:averaging}, see Proposition \ref{prop:averaging}.  The second and third term on the RHS of \eqref{main estimate} come from estimating the second addend on the RHS of \eqref{triangular inequ}, reflecting the fact that $\bar{\theta}_t$ converges to $\theta^{\star}$ for large $t$ (second addend on the RHS of \eqref{main estimate}) and large $\beta$ (third addend on the RHS of \eqref{main estimate}). The latter term is studied by producing a so-called concentration bound. This is dealt with in Appendix B.6 in the Supplementary Material \cite{suppMat}, see Proposition 1 in \cite{suppMat}. Notice that all the bounds we produce here are non-asymptotic (in the relevant parameter) and so they give an explicit way of choosing $\epsilon, \beta$ and $t$ to achieve a given threshold accuracy. 

Let us now make some comments on the significance of estimate \eqref{main estimate}. 
The constant $C$ in front of the first addend on the RHS of \eqref{main estimate} is independent of time; this implies that the convergence, as $\epsilon$ tends to zero, of $\theta_t$ to $\bar{\theta}_t$ ( i.e. the averaging estimate we prove in Section \ref{subsec:averaging}) is UiT. Such an  estimate is  obtained using the approach of \cite{crisan2022poisson}. To explain why it is important to provide UiT convergence in the context of this paper, observe that had we not used the methods of \cite{crisan2022poisson}, other methods of proof,  e.g. \cite{pavliotis2008multiscale}, 
would have resulted in a bound with a constant $C$ which is an increasing function of time (typically exponentially increasing). That is, instead of  \eqref{main estimate} we would have a bound of the following type
\begin{equation}\label{main estimate - non time uniform}
    \left \vert \mathbb E f(\theta_t) - f(\theta^{\star})\right \vert^2\leq \epsilon c e^{ct} (\|\nabla f\|^2_{\infty}+ \|\nabla^2f\|_{\infty}^2) +C_f^2\left(e^{-2\mu t} |{\theta}_0 - {\theta}^\star|^2 +  \frac{2 d_\theta}{\beta} (1 - e^{-2\mu t})\right)\, , 
\end{equation}
for some constant $c>0$. This is problematic because in order for the second addend on the RHS of the above to be `small enough' we need $t$ large; however, if $t$ is large, then the above estimate would lead one to believe that, as $t$ gets larger, we would need to choose a smaller $\epsilon$, with all the complications that this implies from the point of view of numerical implementation, as $\epsilon$ regulates stiffness of the coefficients of the SDE \eqref{eq:theta_update} - \eqref{latent_update}, see Note \ref{note:numerics}. The estimate \eqref{main estimate} is sharper and avoids this issue. 

Moreover,  with similar considerations, \eqref{main estimate}  shows that also the overall convergence, as $\epsilon$ tends to zero and $\beta$ tends to infinity, of $\theta_t$ to $\theta^{\star}$,  is  UiT. This allows for  flexibility in the choice of constants in algorithmic practice (provided the discretization of the SDE for $(\theta_t, X_t)$ preserves this property), in the sense that, when simulating, $\epsilon$ and  $\beta$ can be chosen independent of time and of each other.

 In this paper we have chosen to use the diffusion \eqref{eq:theta_update} - \eqref{latent_update} as a tool to solve the MMLE problem. Equally, we could have used an identical system without the noise component in \eqref{eq:theta_update} which can be viewed as resulting from \eqref{eq:theta_update} - \eqref{latent_update} after having formally let $\beta$ to infinity. We did not do so for two reasons, one is technical, the other less so. Firstly, as we have mentioned,  the theory in \cite{crisan2022poisson}, which we use to produce the UiT averaging bound,  applies to systems of SDEs. This is not to say that it cannot be extended to cover systems of the form where $\theta$-dynamics is deterministic ODE, but the purpose of this paper is not to extend that theory further, here we want to show how it can be applied.  On the less technical side,  we point out that while, under our assumptions, the maximiser $\theta^{\star}$ is unique (owing to our convexity assumption, see Section \ref{sec:sec 4}), this is not necessarily the case. When the maximiser is not unique adding noise to $\theta$-dynamics can be beneficial to explore all the modes, see e.g. \cite{pavliotis2022derivative}. 

{The effect of $\beta$ on the dynamics \eqref{eq:theta_update}-\eqref{latent_update} is well-studied. In the strongly convex setting, which we restrict to here, there is no tradeoff: the larger $\beta$ is, the more concentrated around the unique maximiser (see Proposition 1 in the Supplementary Material \cite{suppMat}). However, in the more general setting one must balance this concentration with sufficient exploration of the potential. That is, if $\beta$ is too high (and, hence, the noise too low), the dynamics may get stuck in local maxima. This is known in the molecular dynamics literature, for a more detailed discussion of this see e.g. \cite{tamd_free_energy_beta, tamd_rare_event, md_free_energy_abrams}. We also include numerics in the setting without strong convexity in Figure \ref{fig:bnn}, indicating that a lower $\beta$ can be beneficial.}

Let us now move on to the numerical approximation of \eqref{eq:theta_update}-\eqref{latent_update}. The proposition below quantifies the discretization error. 

\begin{prop}\label{prop:numerical_error}
Let Assumption \ref{ass:growthcoeffs} and Assumption \ref{assmp:convexity} (stated in Section \ref{subsec:assumptions})  hold,  and furthermore assume the potential $U$ satisfies the following: 
\begin{align}\label{ass:euler-estimate}
    \left\vert\nabla U(z) \right\vert^2 \geq c + \tilde{c} \, |U(z)|
\end{align}
for some $c\geq 0, \tilde{c}> 0 $.  Let $\zz_t=(\theta_t, X_t)$ be the solution of the SDE  \eqref{eq:theta_update}- \eqref{latent_update} and  let $\zz_k^{\delta} = (\theta_k^{\delta}, X_k^{\delta})$ be its Euler-Maruyama discretization,  as in Algorithm \ref{alg:ais}. Then, for each $\epsilon, \beta>0$ fixed,  there exists $\xi \in (0, 1/2)$ and positive constants $\tilde{\lambda}=\tilde{\lambda}(\delta)$, $\lambda = \lambda(\epsilon)$, $G=G(\epsilon, \beta)$, $\tilde G=\tilde G(\epsilon, \beta)$ and $\tilde C= \tilde C(\epsilon, \beta, \delta)$ (with $\tilde C$ tending to zero as $\delta$ tends to zero, for each $\beta, \epsilon$ fixed)
such that the following holds 
\begin{align}\label{bound-numerical-error}
\left\vert \mathbb E g(\zz^{\delta}_k) - \mathbb E g(\z_{t})\right \vert \leq 
 \tilde C U(\z_0)  e^{-\tilde{\lambda}  \, k \delta} + G \delta^{\xi}  + \tilde{G} e^{-\lambda t} (1+ U(\z_0))
\end{align}
where $t=k\delta$,  $\z_0=(X_0,\theta_0)$ is the initial datum for both the SDE \eqref{eq:theta_update}- \eqref{latent_update}, and its discretization.  The above holds for any smooth and bounded function $g$. \footnote{In reality under our assumptions it holds for every $g$ which is at most quadratically growing.}
\end{prop}

\begin{note}\label{note:numerics}
    \textup{Some comments on the statement of Proposition \ref{prop:numerical_error}. 
    \begin{itemize}
    \item  In our numerical experiments of the next section we use the Euler-Maruyama scheme in order to simulate the dynamics \eqref{eq:theta_update}- \eqref{latent_update}. We do this for simplicity and because this was sufficient for the test cases we present here. However note that, since $\epsilon$ is small, equation \eqref{latent_update} has stiff coefficients and  it is well known that using the Euler-Maruyama discretization for this type of equations is quite expensive -- the scaling of $\epsilon$ with $\delta$ is inversely proportional, i.e. the smaller $\epsilon$, the smaller the step-size $\delta$ that one needs to choose, see \cite[Chapter 9]{weinan2011principles} for a precise account on this. So our choice of numerical scheme is by far not optimal. Since this choice is not optimal we do not carry out a refined analysis of the approximation error. In particular it is important to point out that the constant  $\tilde C$ in \eqref{bound-numerical-error} depends on $\epsilon$ (and $\beta$). We do not track the exact dependence of $ \tilde C$ on such parameters in our proof. From numerical experiments and because of the stiffness of the coefficients, we know that  this constant will be at least inversely proportional to $\epsilon$, i.e. $\tilde C \sim \epsilon^{-\ell}$,  for some  $\ell\geq 1$. 
    Bottom line, the bound \eqref{bound-numerical-error}  implies that  for each $\epsilon$ fixed, the Euler discretization $\zz^{\delta}_k$ tends to $\zz_t$, as $\delta$ tends to zero {\em and} $t\rightarrow \infty$ (which is what we need). \footnote{That is, it is an asymptotic result,  and does not imply the convergence of the discretization $\z_k^{\delta}$ to $\z_t$ for each $t$ fixed, as $\delta$ tends to zero. This is in contrast with bounds of the form \eqref{uitdefinition}, which are stronger as they are non-asymptotic and hence also  imply  convergence of the finite time marginals. }  However  $\delta$ needs to be scaled with $\epsilon$ if we want the RHS of  \eqref{bound-numerical-error} to go to zero  as $\epsilon$ tends to zero.
    \item  The proof of  \eqref{bound-numerical-error} is done by applying  the results of \cite{mattingly2002ergodicity}.  In that scheme of proof it would  be very laborious to track the dependence of the various constants on $\epsilon$ and $\beta$. In particular from that proof it is difficult to know how $G$ and $\tilde G$ depend on the various parameters. A more detailed analysis of the constant $\tilde{G}$ has been carried out in \cite[Theorem 2]{stoltz2018longtime} (more comments on this in the next section).    Another route  would have been to use the methods of proof in \cite{angeli2023uniform}, which have been designed to produce UiT convergence results for numerical approximations. However the assumptions needed to employ those methods are not satisfied by the overall drift of the system \eqref{eq:theta_update}- \eqref{latent_update}, so those schemes of proof (which in fact use the same principles that we have employed here for our UiT averaging result, i.e. they also leverage the concept of Strong Exponential Stability) would need to be non-trivially extended. \footnote{When producing our UiT averaging result we will only need Strong Exponential Stability (SES) of the FrP \eqref{eqn:frozen} and of the averaged process \eqref{eq:theta_bar_sde}, see Lemma  \ref{lemma:fastsemigroupderests} and Lemma \ref{lem:avgderest}, respectively. If we were to use the same general approach to prove a UiT estimate for the Euler discretization then we would need SES of the (semigroup associated with) system \eqref{eq:theta_update}- \eqref{latent_update} as well. Unsurprisingly, this is where the difficulty arises, in the small $\epsilon$ regime - indeed, this would not be a problem at all for large $\epsilon$. } As we said, we don't do it here because we know that the Euler-Maruyama scheme is anyway suboptimal. 
    \end{itemize}
    }
\end{note}

Finally, putting together Theorem \ref{thm:main_thm} and Proposition \ref{prop:numerical_error}, we obtain the following. 

\begin{thm}\label{thm:combined result}
Let $\theta_k^{\delta}$ be the Euler-Maruyama discretization of \eqref{eq:theta_update},  as in Algorithm \ref{alg:ais} and let $\theta^{\star}$ be as in \eqref{eq:mmle}. Suppose the assumptions of Theorem \ref{thm:main_thm} and Proposition \ref{prop:numerical_error} hold. Then the following estimate holds
\begin{align*}
\left \vert\mathbb E f(\theta_k^{\delta}) - f(\theta^{\star})\right \vert^2 & \leq C \epsilon  + C_f^2\left(e^{-2 \mu t } |\theta_0-\theta^{\star}|^2+ \frac{2 d_{\theta}}{\beta} (1-e^{-2 \mu t })\right)\\ 
& +\tilde{C} U(z_0) e^{-\tilde \lambda k \delta} + G \delta^{\xi}+ G e^{-\lambda t} (1+ U(z_0)) \, ,
\end{align*}
where $C$ is a constant independent of $\epsilon, t, \delta, \beta$ and all the other constants and parameters are as in the statements of Theorem \ref{thm:main_thm} and Proposition \ref{prop:numerical_error}. 
\end{thm}
The above result guarantees convergence of the Euler discretization $\theta_k^{\delta}$ to $\theta^{\star}$.

\subsection{Relation to literature}\label{sec:relevant_work}

Slow-fast systems of SDEs have already been used in optimization, in Bayesian inverse problems and  for MMLE.  We mention here the works which are most related in spirit to the present work, without any claim to completeness.  
An interacting system of SDEs has been used in \cite{pavliotis2022derivative}, in the context of Bayesian inversion (see also references within that work, which are extensive). The work \cite{pavliotis2022derivative} is also based on a Langevin diffusion, not dissimilar to   the one we are using here, equation \eqref{eq:theta_bar_sde}. However in \cite{pavliotis2022derivative} the underlying statistical model does not incorporate latent variables; that is, the marginal likelihood of interest, $p_{\theta}(y)$,  is not specifically seen as a marginal with respect to latent variables, so the fast variables are not used, as we do here, to evaluate the conditional distribution $p_{\theta}(x\vert y)$.  In \cite{pavliotis2022derivative} the fast variables (which are different from the ones in this paper) are instead used for local exploration in parameter space.  Multiscale dynamics have also been used in the context of neural networks \cite{chaudhari2018deep, kantas2019sharp},  to smoothen the loss function, and in controlled, non-asymptotic versions of annealing-like procedures \cite{breiten2021stochastic}. 

As we have already mentioned, more recently, \cite{kuntz2023particle} studied an {interacting} particle system designed for the same purpose of solving the MMLE problem,  and proposed an algorithm called particle gradient descent (PGD). This algorithm operates $N$  particles to estimate the expectation, which are propagated in parallel, and thus can be seen as a space analogue of running a Markov chain. While \cite{kuntz2023particle} provided a limited analysis of the continuous-time gradient flow that underlies the proposed SDE, a slightly modified version of this algorithm was proposed by \cite{akyildiz2023interacting}, termed the interacting particle Langevin algorithm (IPLA). The work \cite{akyildiz2023interacting} provided a full analysis of IPLA in Wasserstein-2 distance, imposing strong convexity and global Lipschitz conditions negative joint log-likelihood.

As we have already mentioned in the introduction, the algorithm we propose here is a reinterpretation, in statistical context, of algorithms which have been used for a long time in molecular dynamics, such as  HMM \cite[Chapter 9]{weinan2011principles}, EFC \cite{kevrekidis2009equation} or TAMD \cite{stoltz2018longtime}. Some theoretical guarantees have already been provided for such methods, see again \cite{weinan2011principles} and references therein. However, on the theoretical front, the paper which is closest to the spirit to the results in this paper is \cite{stoltz2018longtime}, which, aside from the already mentioned \cite{crisan2022poisson} and the present paper,  contains  the only other  averaging result on long time horizons (and for systems of SDEs) that we are aware of.  To explain the difference between the averaging result produced to prove
 Theorem \ref{thm:main_thm}  of this paper, namely Proposition  \ref{prop:averaging} (see next section) and the averaging result in  \cite{stoltz2018longtime}  let us observe that the multiscale dynamics $(\theta_t, X_t)$ in \eqref{eq:theta_update}- \eqref{latent_update} has, for each $\epsilon$ and $\beta$ fixed,  an invariant measure, let us call such a measure $\rho^{\epsilon, \beta}$, which is a measure on $\R^{d_{\theta}}\times \R^{d_x}$. When $\epsilon$ is small, it is natural to expect that the $x-$ marginal of this measure, i.e. the measure  $\bar{\rho}^{\epsilon, \beta} = \int_{\R^{d_x}} dx \, \rho^{\epsilon, \beta}(\theta, x)$, is close to the measure $\bar{\rho}_{\beta}$ introduced in \eqref{eqn:rho-beta}, which is the invariant measure of the averaged process $\bar{\theta}_t$ (solution of \eqref{eq:averaged_theta_bar}). Hence in \cite{stoltz2018longtime} the authors proceed  by first showing that $(\theta_t, X_t)$ converges, as $t \rightarrow \infty$, and for each $\epsilon, \beta$ fixed, to $\bar{\rho}^{\epsilon, \beta}$ (with a rate of convergence which is shown to be independent of $\epsilon$). They then show, through a careful and very nice perturbative expansion in $\epsilon$ of the measure $\bar{\rho}^{\epsilon, \beta}$, that $\bar{\rho}^{\epsilon, \beta}$ is close to $\bar{\rho}^{\beta}$, when $\epsilon$ is small.  So,  in the analysis of \cite{stoltz2018longtime}, the overall convergence of \eqref{eq:theta_update} to $\bar \rho^{\beta}$ is proved by combining the stability property of  \eqref{eq:theta_update}-\eqref{latent_update} to the analysis of the asymptotic bias.  Hence, the estimate in  \cite{stoltz2018longtime} results in a bound of the form
 \begin{equation}\label{eqn:gabriel}
 \|  \theta_t - \bar{\rho}^{\beta}\|_{*} \leq c\epsilon + ce^{-Kt}
 \end{equation}
 where $\| \cdot \|_{*}$ is an appropriate norm (for detail see \cite{stoltz2018longtime}). The above bound is weaker than a bound of the form \eqref{uitdefinition} (and this does not depend on the fact that in \cite{stoltz2018longtime} the authors use a norm different from ours, it depends on the fact that the bound \eqref{uitdefinition} is non-asymptotic and it implies also convergence of the finite time marginals, the above bound does not.) Moreover \cite{stoltz2018longtime} does not contain the analysis of the discretization error. However it does accommodate for the case when the slow and fast dynamics evolve at different temperatures, which is something we don't do here, as we don't need it for our purposes. 
  For further comparison between these works we refer the reader to  \cite{crisan2022poisson}.  

We now come to compare the algorithm we have presented in this paper with SOUL, and explain why SOUL can be interpreted as a two-scale algorithm for MMLE.

\subsubsection{Relation to SOUL} \label{sec:comparison_SOUL}
The SOUL algorithm \cite{de2021efficient}  aims at explicitly implementing the EM algorithm by alternating inexact MCMC methods (such as unadjusted Langevin algorithm) for realising E-step in \eqref{eq:theEM-E-Step} and gradient-based optimisation for realising the M-step in \eqref{eq:theEM-M-Step}. From our perspective, the SOUL algorithm is implemented to realise the averaging procedure explicitly: To update $\theta_k$ at iteration $k$, a separate, long Markov chain $X_k^{(m)}$ with $m = 0, \ldots, M$ is run, with the assumption that this Markov chain (e.g. unadjusted Langevin algorithm) would converge to a measure that is close to the stationary measure of the chain in $M$ iterations. These MCMC samples are then plugged into gradient estimates to update $\theta$. We recall the EM algorithm steps as given in \eqref{eq:theEM-E-Step}--\eqref{eq:theEM-M-Step} which can be succinctly described iteratively solving $\theta_{k} \in \arg \min_{\theta\in \mathbb{R}^{d_\theta}} \mathbb{E}_{p_{\theta_{k-1}}(x|y)} \left[ U(\theta, x)\right]$. Therefore, given a starting point $\theta_{k-1}$, the main idea behind SOUL is to first implement a ULA-type algorithm to sample from $p_{\theta_k}(x |y) \propto p_{\theta_k}(x, y)$. For fixed $\theta_{k-1}$, the SOUL algorithm proceeds running a ULA chain with a warm start $X_{k}^{(0)} = X_{k-1}^{(M)}$, as follows:
\begin{align}\label{eqn:soul_latent_update}
X^{(m)}_{k} = X^{(m-1)}_{k} - \latentstep \nabla_x U(\theta_{k-1}, X_{k}^{(m-1)}) + \sqrt{2\latentstep}W_{k+1}
\end{align}
for $m = 0, \ldots, M$. This chain should be run sufficiently long for samples to be approximately from the stationary distribution, that is, $p_{\theta_{k-1}}(x|y)$ (which is the stationary measure of the \textit{frozen} process in our context). Given these samples from $p_{\theta_{k-1}}(x|y)$, in order to maximise the expectation $\mathbb{E}_{p_{\theta_{k-1}}(x|y)} \left[ U(\theta, x)\right]$, the last $\tilde{M}$ samples (after burn-in) are used to estimate the gradient: 
\begin{align}\label{eqn:soul_theta_update}
\theta_k = \theta_{k-1} - \frac{\thetastep}{\tilde{M}} \sum_{m=1}^{\tilde{M}} \nabla_{\theta} U\left(\theta_{k-1}, X_t^{(M-\tilde{M}+m)}\right).
\end{align}
One can see that this algorithm, in some sense, takes the averaging \textit{literally}: For every parameter update, an explicit and long Markov chain is simulated, whose samples are then used to estimate the gradient for optimisation of $\theta$. Our approach is to have the same effect without the need of running Markov chains for every update of $\theta$, via the use of stochastic averaging, which is the natural framework to implement such algorithms.
\subsubsection{Relation to interacting particle algorithms} \label{sec:comparison_IPLA}
Another relevant class of methods, as mentioned in the introduction, are the ones based on \textit{interacting particle systems}, namely PGD \cite{kuntz2023particle} and IPLA \cite{akyildiz2023interacting}.

The work \cite{kuntz2023particle} developed a formulation based on gradient flows, which resulted in the system of SDEs (in our notation)
\begin{align}
    \md {\theta}^N_t &= -\frac{1}{N}\sum_{j=1}^N \nabla_{\theta} U({\theta}^N_t, {X}_t^{j, N})\md t 
 \label{eq:ContIPS_theta_kuntz} \\
    \md {X}_t^{i, N} &= -\nabla_x U({\theta}^N_t, {X}_t^{i, N})\md t + \sqrt{2}\md {B}_t^{i, N},\label{eq:ContIPS_x_kuntz}
\end{align}
for $i = 1, \ldots, N$ and proposed its Euler-Maruyama discretisation as a practical algorithm, called particle gradient descent (PGD). This can be seen as direct ``space analogue'' of SOUL, as PGD algorithm, instead of running separate chains for every $\theta_t$, retains $N$ particles in parallel. A nonasymptotic analysis of this algorithm was recently provided in \cite{caprio2024error} for the strongly convex setting. Following the approach in \cite{kuntz2023particle}, \cite{akyildiz2023interacting} modified this system of SDEs slightly to arrive at
\begin{align}
    \md {\theta}^N_t &= -\frac{1}{N}\sum_{j=1}^N \nabla_{\theta} U({\theta}^N_t, {X}_t^{j, N})\md t + \sqrt{\frac{2}{N}}\md {B}_t^{0, N}
 \label{eq:ContIPS_theta_aky} \\
    \md {X}_t^{i, N} &= -\nabla_x U({\theta}^N_t, {X}_t^{i, N})\md t + \sqrt{2}\md {B}_t^{i, N},\label{eq:ContIPS_x_aky}
\end{align}
for $i = 1, \ldots, N$, by adding scaled noise to $\theta$-dimension. The discretisation of this set of SDEs resulted in an algorithm named Interacting Particle Langevin Algorithm (IPLA). \cite{akyildiz2023interacting} showed that the stationary distribution of the system \eqref{eq:ContIPS_theta_aky}--\eqref{eq:ContIPS_x_aky} concentrates on $\theta^\star$ (true maximum marginal likelihood estimate) as $N \to \infty$, which plays a similar role to the inverse temperature parameter $\beta$ in our setting. However, naturally, the concentration requires large $N$ which translates into simulating more particles within the setting of PGD and IPLA. 

We should note however that while our bounds are given in weak error, \cite{akyildiz2023interacting} provide bounds in strong error. Similarly, for the PGD, within the strongly convex setting, \cite{caprio2024error} proved their bounds in strong error.

Parallels between these methods (specifically IPLA, as it also injects stochasticity into $\theta$-updates) can be drawn, as we can also reinterpret \eqref{eq:ContIPS_theta_aky}--\eqref{eq:ContIPS_x_aky} within a stochastic averaging framework (a full study is beyond the scope of the current paper). Intuitively, notice that for fixed $\theta_t^N$, the particle system in \eqref{eq:ContIPS_x_aky} for $i = 1, \ldots, N$ have the stationary product measure (as all dimensions are independent for frozen $\theta$ dimension) $p(x_1, \ldots, x_N | y) = \prod_{i=1}^N p(x_i | y)$. {Thus, if we accelerate the SDE \eqref{eq:ContIPS_x_aky} using $\varepsilon$ in the usual way}, we can see that the \textit{averaged} dynamics given the stationary measure of the collection of latent variable processes takes the form
\begin{align}
    \md {\theta}^N_t &= -\frac{1}{N}\sum_{j=1}^N \left( \int \nabla_{\theta} U({\theta}^N_t, x_i) p(x_i | y) \md x_i \right) \md t + \sqrt{\frac{2}{N}}\md {B}_t^{0, N}
\end{align}
which simplifies exactly to the averaged process in \eqref{eq:theta_bar_sde} with $\beta = N$. We leave more rigorous extension of our work into this direction for future work.

Finally, we also mention the \textit{implicit diffusion} approach \cite{marion2025implicit}. This approach is similar to the particle-based approaches presented above, to optimise general cost functions with a two time-scale system. Specialised to the MMLE problem (and rewritten with fixed $\epsilon$), the approach in \cite{marion2025implicit} results in a system of (idealised) SDEs
\begin{align}
\md \theta_t &= -\epsilon \int \nabla_{\theta} U(\theta_t, x) p_{\theta_t}(x|y) \md x \, \md t \label{eq:implicit-diff-theta} \\
\md X_t &= -\nabla_x U(\theta_t, X_t) \md t + \sqrt{2} \md W_t \label{eq:implicit-diff-x}
\end{align}
The algorithm implements a time-discretisation of the above SDEs (with a particle approach to approximate \eqref{eq:implicit-diff-theta} in space), which is identical to the PGD algorithm \citep{kuntz2023particle} when $\epsilon = 1$. We make two general remarks about the relationship our work and this method: First, the work of \citet{marion2025implicit} provides a bound that vanishes as $t \to \infty$, in particular under their assumptions (which are different from ours), they obtain a bound of the form $\mathcal{O}(1/t\epsilon + 1/t + \epsilon)$. This bound is asymptotically $\mathcal{O}(\epsilon)$ as in our case in Theorem~\ref{thm:main_thm}. However, unlike our continuous-time result, this bound does not hold, for every fixed $t$, as $\epsilon \to 0$ (i.e., it is not uniform in time in the usual sense referenced in averaging literature, see, e.g., \citet{crisan2022poisson}). Second, the main discretized result obtained in \cite{marion2025implicit} depends upon use of the law of $X_k$ (hence it is just a time-discretisation of \eqref{eq:implicit-diff-theta}-\eqref{eq:implicit-diff-x} without further space discretization), which in practice is not available. The further discretization error induced when approximating the law of $X_k$ is not studied in \cite{marion2025implicit}. The discretization used here is an Euler-Maruyama scheme on the continuous slow-fast dynamics, and not dependent on knowledge of the law.

\section{Statement of assumptions and proofs}\label{sec:sec 4}
The proofs of all results can be found in the Supplementary Material \cite{suppMat}. In this section we first state our assumptions, before we state the results needed for Theorem \ref{thm:main_thm} and Proposition \ref{prop:numerical_error}. After stating our assumptions, we split our analysis into three main sections. In the first section, Section \ref{subsec:averaging}, we state Proposition \ref{prop:averaging}, a uniform-in-time averaging bound between \eqref{eq:theta_bar_sde} and \eqref{eq:theta_update}, using the results in Section~\ref{sec:sa_background}. Secondly, in Appendix B.6 of the Supplementary Material \cite{suppMat} we provide a concentration result, Proposition 1, on the stationary measure of \eqref{eq:theta_bar_sde} around the MMLE. Proposition \ref{prop:averaging} here and Proposition \ref{prop:concentration bound} from \cite{suppMat} combine to give Theorem \ref{thm:main_thm}. Let us first list our assumptions. 
\subsection{Assumptions}\label{subsec:assumptions}
In what follows, we  first list our assumptions and then comment on them, in turn. Throughout, $|\cdot|$ is the euclidean norm if the argument is vector-valued and the Frobenius norm\footnote{That is, the norm $\|A\|_F^2 \coloneqq \sum^{n}_{i=1}\sum^{d}_{j=1} (A_{ij})^2$ for $A \in \R^{n\times d}$} if it is matrix-valued. We make a standing assumption, and we don't repeat this in every statement,  that $U \in \mathcal{C}^3(\R^{d_x}\times \R^{d_\theta})$. That is, all partial derivatives of $U$ up to order 3 exist and are continuous. 

\begin{assumption}\label{ass:growthcoeffs} (Growth of coefficients). The function $U$ satisfies the following conditions: 
	\begin{enumerate}
		\item 
		There exists a constant $L>0$ such that
		\begin{equation}
		|\nabla U(\theta,x) - \nabla U(\theta',x')| \leq L (|\theta-\theta'|+|x-x'|),\end{equation}
		and some constant $C>0$, independent of $\theta$, such that
		\begin{equation} |\nabla_x U(\theta,x)| \leq C(1+|x|),
		\end{equation}
		for all $x,x' \in \R^{d_x}$ and $\theta,\theta' \in \R^{d_\theta}$.
		In particular this implies that there exists a constant $D\geq 0$ (independent of $x\in \R^{d_x}, \theta \in \R^{d_\theta}$) such that \begin{equation}\label{eqn:secderbound}
		    |\nabla_x \partial_{\theta_i}U(\theta,x)|^2 \leq D
		\end{equation} for all $x \in \R^{d_x}, \theta \in \R^{d_\theta}$ and $1\leq i \leq d_\theta$.
		\item \label{item:bdd3rdDeriv} There exists $C\geq 0$ (again, independent of $x\in \R^{d_x}, \theta \in \R^{d_\theta}$) such that \begin{equation}\label{eqn:thirdderbound}|\nabla^3 U(\theta,x)| \leq C.
        \end{equation}
	\end{enumerate}
\end{assumption}

\begin{assumption}\label{assmp:convexity} Let $z = (\theta, x)$. There exists a constant $\mu > 0$ such that
\begin{align*}
\langle z - z', \nabla_z U(z) - \nabla_z U(z') \rangle \geq \mu |z - z'|^2,
\end{align*}
for all $z, z' \in \R^{d_\theta + d_x}$, where $\langle \cdot \rangle$ denotes Euclidean scalar product.
\end{assumption}

\begin{assumption}\label{ass:strongconvexfast} (Strong convexity in $x$ of the FrP \eqref{eqn:frozen}). There exists a positive constant $\kappa>0$ such that for any $\theta\in \mathbb{R}^{d_\theta}, x\in \mathbb{R}^{d_x}, \xi\in \mathbb{R}^{d_x}$ we have
	\begin{equation}\label{eqn:fastStrongConvex}
	\begin{aligned}
	&-2 \sum_{i,j=1}^n \partial_{x_i} \partial_{x_j} U(\theta,x) \xi_i\xi_j \leq -\kappa \lvert\xi\rvert^2.
	\end{aligned}
	\end{equation}
	
\end{assumption}

\begin{assumption}\label{ass:avgderest} (Strong convexity in $\theta$ of the averaged equation).
	There exists $\zeta_0>0$ independent of $\theta,x,\xi$ such that for any $\theta\in \mathbb{R}^{d_\theta}, x\in \mathbb{R}^{d_x}, \xi\in \mathbb{R}^{d_\theta}$
	\begin{equation}\label{eqn:slowstrongconvex1}
	\begin{aligned}
	&- \sum_{i,j=1}^n \partial_{\theta_i} \partial_{\theta_j} U(\theta,x) \xi_i\xi_j \leq -\left(\frac{d_\theta D^2}{\kappa}+\zeta_0\right) \lvert\xi\rvert^2,
	\end{aligned}
	\end{equation}
	with $D, \kappa$ as in Assumption \ref{ass:growthcoeffs} and Assumption \ref{ass:strongconvexfast} respectively.
\end{assumption}
Examples of functions $U$ satisfying all of the above assumptions are given below.
\begin{example}
We consider the following function $U$, with $d_x = d_\theta = d$:
    \begin{equation}\label{eqn:ex1U}
        U(x,\theta) = a|x|^2 + \tilde{a}|\theta|^2 + \sum^d_{i=1}b(\theta_i)g(x_i),
    \end{equation} where $a, \tilde{a}$ are constants, $a, \tilde{a} >0$ and $b,g \in \mathcal{C}_b^3(\R)$ are smooth functions (bounded and with bounded derivatives up to order $3$). Suppose there exists $\kappa>0$ such that
        \begin{equation*}
        2a-\|g''\|_\infty\|b\|_\infty \geq \kappa > 0.
    \end{equation*} Then if
    \begin{equation*}
        2\tilde{a}-\|b''\|_\infty\|g\|_\infty \geq \frac{d \|g'\|^2_\infty \|b'\|^2_\infty}{2\kappa},
    \end{equation*}
    $U$ as in \eqref{eqn:ex1U} satisfies all assumptions.
\end{example}

\begin{example}
We consider the following function $U$, again, with $d_x = d_\theta=d$:
    \begin{equation}\label{eqn:ex2U}
        U(x,\theta) = a|x|^2 + \tilde{a}|\theta|^2 + \sum^d_{i=1}h(x_i - \theta_i),
    \end{equation} where $a, \tilde{a}$ are constants, $a, \tilde{a} >0$ and $h \in \mathcal{C}_b^3(\R)$ is a smooth function (bounded and with bounded derivatives up to order $3$). Suppose there exists $\kappa>0$ such that
        \begin{equation*}
        2a-\|h''\|_\infty \geq \kappa > 0.
    \end{equation*} Then if
    \begin{equation*}
        2\tilde{a} \geq \left(1+\frac{d}{2\kappa}\right)\|h''\|_\infty ,
    \end{equation*}
    $U$ as in \eqref{eqn:ex2U} satisfies all assumptions.
\end{example}

Assumption \ref{assmp:convexity}, Assumption \ref{ass:strongconvexfast}
 and Assumption \ref{ass:avgderest} are all strong-convexity type assumptions but they do not imply each other. To clarify, let us make some comments on them. 
 Assumption \ref{assmp:convexity}  enforces strong convexity,  jointly in the variables $x$ and $\theta$. This assumption is used to obtain the concentration result of Proposition \ref{prop:concentration bound} in the Supplementary Material \cite{suppMat}.
 
 Assumption \ref{ass:strongconvexfast} implies the following 
\begin{align*}
\langle x - x', \nabla_x U(x,\theta) - \nabla_x U(x', \theta) \rangle \geq \frac{\kappa}{2} |x - x'|^2 \, , \quad \text{for all } x, x' \in \R^{d_x}, \, \theta \in \R^{d_\theta}
\end{align*}
It can be  therefore interpreted as a strong convexity assumption in the variable $x$, uniformly in $\theta$. Similarly, Assumption \ref{ass:avgderest} is a strong convexity assumption in the variable $\theta$, uniformly in $x$. These two together do not imply joint strong convexity, so Assumption \ref{ass:strongconvexfast}  and Assumption \ref{ass:avgderest} do not imply Assumption \ref{assmp:convexity}.  {Moreover we observe that the constants appearing in Assumption \ref{ass:strongconvexfast} and Assumption \ref{ass:avgderest} are important (in particular, $\kappa$ appears in both). This means Assumption \ref{assmp:convexity} does not imply Assumption \ref{ass:strongconvexfast}  and Assumption \ref{ass:avgderest}, at least without enforcing $\mu = \max\{\kappa/2, \frac{d_\theta D^2}{\kappa}+\zeta_0\}$, which is not sharp. Because Assumption  \ref{assmp:convexity} is only needed for Proposition \ref{prop:concentration bound} in the Supplementary Material \cite{suppMat}, we keep the constant $\mu$ free.}

The reason why the constants appearing in Assumption \ref{ass:strongconvexfast}  and Assumption \ref{ass:avgderest}  are related is not just technical; this is indeed needed to prove our averaging result Proposition \ref{prop:averaging}.

In Note \ref{note:polygrowth} (see Section \ref{subsec:averaging}) we see how Assumption \ref{ass:growthcoeffs} may be relaxed with some work, see Note \ref{note:polygrowth} below, and then, in Lemma \ref{lem:strongmonotonicity}, we show a couple of consequences of our strong convexity conditions, Assumption \ref{ass:strongconvexfast} and Assumption \ref{ass:avgderest}, namely that they imply strong monotonicity conditions on the drift coefficients of \eqref{eq:theta_update}-\eqref{latent_update}. This fact  will be explicitly used in the proofs.  

\begin{lem}\label{lem:strongmonotonicity}
	Let Assumption \ref{ass:growthcoeffs} and Assumption \ref{ass:strongconvexfast} hold. Then there exist constants $r,R > 0$ such that
	\begin{equation}\label{eqn:stronglymonotonicx}
	-\langle \nabla_x U(\theta,x), x\rangle\leq -r|x|^2+R
	\end{equation}
	for every $x,\theta$. Furthermore if Assumption \ref{ass:avgderest} holds then there exist constants $\tilde{r},\tilde{R} > 0$ such that
	\begin{equation}\label{eqn:stronglymonotonic2}
	-\langle \nabla_\theta U(\theta,x), \theta\rangle\leq -\tilde{r}|\theta|^2 + \tilde{R}
	\end{equation}
	for every $x,\theta$.
\end{lem}
\begin{proof}
The proof of Lemma \ref{lem:strongmonotonicity} can be seen in \cite{suppMat}.
\end{proof}

\subsection{Uniform in time averaging estimate}\label{subsec:averaging}

In this section we show that the process \eqref{eq:theta_update} is a UiT approximation, as $\epsilon \rightarrow 0$, of the process \eqref{eq:theta_bar_sde}. This is stated in the main result of this section, Proposition \ref{prop:averaging}. We will be using Theorem 3.3 from \cite{crisan2022poisson}, so most of the work here will involve verifying the assumptions of such a theorem.  For the reader's convenience,  shortly after stating 
Proposition \ref{prop:averaging}, we recall \cite[Theorem 3.3]{crisan2022poisson}, using a notation which is more adapted to the present work. 

The conditions of \cite[Theorem 3.3]{crisan2022poisson} that will require the most effort to establish are  \cite[Equation (45) and (46)]{crisan2022poisson}, which are  semigroup derivative estimates on the FrP \eqref{eqn:frozen} and the averaged process \eqref{eq:theta_bar_sde} respectively. These are proved in   Lemma \ref{lemma:fastsemigroupderests} and Lemma \ref{lem:avgderest} below, which largely follow the scheme of proof of \cite[Lemma 2.8]{angeli2023uniform} and \cite[Lemma 2.10]{angeli2023uniform}, but in the case of constant diffusion coefficients. When the diffusion coefficients are constant there are some simplifications we can make here, and this will be commented on in the proofs of Lemma \ref{lemma:fastsemigroupderests} and Lemma \ref{lem:avgderest} themselves. In the following, and throughout, we use the notation $\|g\|^2_\infty = \sup_{(x, \theta) \in \R^{d_x}\times  \R^{d_\theta}} \left|g(x,\theta)\right|^2$, for any vector valued function $g$ on $\R^{d_x}\times \R^{d_\theta}$. 
\begin{prop}\label{prop:averaging}
	{Let Assumption \ref{ass:growthcoeffs}-\ref{ass:avgderest} hold.} Let $\theta_t^{\epsilon,\beta,\theta,x}$ be as in \eqref{eq:theta_update}, and $\bar{\theta}^\beta_t$ be as in \eqref{eq:theta_bar_sde}, here the notation $\theta_t^{\epsilon,\beta,\theta,x}$ and $\bar{\theta}^\beta_t$ is to indicate explicitly the dependence of the solutions to \eqref{eq:theta_update} and \eqref{eq:theta_bar_sde} on both the parameters $\beta,\epsilon$ as well as the initialisation data $x,\theta$. Then there exists $C_0>0$, independent of $\beta, x, \theta$ such that for any $f \in \mathcal{C}_b^2(\R^{d_\theta})$ the following holds
	\begin{equation}\label{eqn:uitAveraging}
	\left \vert
	\mathbb E f(\theta_t^{\epsilon,\beta,  \theta, x}) - \mathbb E f(\bar{ \theta}^{\beta,\theta}_t)\right \vert \leq \epsilon C\left(\left\|\nabla f\right\|_\infty^2 +  \left\|\nabla^2 f\right\|_\infty^2\right),  
	\end{equation} with $C = C_0(1+|x|^2 + |\theta|^2)$.
\end{prop}
\begin{proof}
    The proof of Proposition~\ref{prop:averaging} can be found in the Supplementary Material \cite{suppMat}.
\end{proof}

 In the following, and in the proofs in the Supplementary Material \cite{suppMat}, we use semigroup notation extensively, so we define it here. We use $P_t^{\theta}$ to denote the semigroup associated with the FrP \eqref{eqn:frozen}, and we use $\bar{P}^\beta_t$ to denote the semigroup associated with the averaged process \eqref{eq:theta_bar_sde}. That is,
\begin{equation}
    \left(P_t^{\theta}f\right)(x) = \mathbb{E}\big[f( X^\theta_t )\big],\quad f \in \mathcal{C}^2_b(\R^{d_x}),
\end{equation} and
\begin{equation}
    \left(\bar{P}^\beta_t f\right)(\theta) = \mathbb{E}\big[f( \bar{ \theta}^{\beta,\theta}_t )\big], \quad f \in \mathcal{C}^2_b(\R^{d_\theta}).
\end{equation}

 As we have already pointed out, in order to prove Proposition \ref{prop:averaging} we will heavily use \cite[Theorem 3.3]{crisan2022poisson}. For the reader's  convenience we restate it below. 


Consider a slow-fast  system of SDEs of the form 
\begin{align}
    d\theta_t & = b(\theta_t, X_t) dt+ \sqrt{2} \tilde{b}(\theta_t) dW_t^0 \label{eqn:slowgeneral thm}\\
    dX_t & =\frac{1}{\epsilon} a(\theta_t, X_t) dt + \sqrt{\frac{1}{\epsilon}} \tilde{a}(\theta_t) dW_t^1 \, \label{eqn:fast general thm}
\end{align}
with  $(\theta_t, X_t) \in \R^{d_{\theta}} \times \R^{d_x}$, $b, W^0$ and $a, W^1$ as in \eqref{eqn:slowintro}- \eqref{eqn:fastintro} and $\tilde{b}:\R^{d_{\theta}}\rightarrow \R^{d_{\theta} \times d_{\theta}}$, $\tilde{a}: d_{\theta} \rightarrow \R^{d_x \times d_x}$.  
The associated frozen process is given by
\begin{equation}\label{eqn:frozengeneralthem}
    dX_t^{\theta} = a(\theta, X_t^{\theta}) dt + \tilde{a}(\theta) dW^1_t \, .
\end{equation}
{
Under well-known conditions, for every $\theta$ fixed the frozen process \eqref{eqn:frozengeneralthem}  admits a unique invariant measure $\mu^{\theta}$ (which is a probability measure on $\R^{d_x}$). Hence we can consider the process
\begin{equation}\label{eqn:avggeneralthem}
d\bar \theta_t= \bar b(\bar \theta_t) dt + \sqrt{2} \tilde b(\bar\theta_t) dW^0_t \, , 
\end{equation}
where $\bar b (\theta) = \int_{\R^{d_x}} b(\theta, x) \mu^{\theta}(dx)$. We denote by $\mathcal{P}^{\theta}_t$ the semigroup associated to \eqref{eqn:frozengeneralthem}, and by $\bar{\cP}_t$ the semigroup associated to \eqref{eqn:avggeneralthem}.}

Proposition \ref{thm3.3} is split into two parts. Proposition \ref{thm3.3} \ref{propfirstversion} is a version of \cite[Theorem 3.3]{crisan2022poisson}, rewritten in the notation of this paper for the reader's convenience. Proposition \ref{thm3.3} \ref{propsecondversion} is a simplified version of \cite[Theorem 3.3]{crisan2022poisson}, in the case of constant diffusion and drift coefficients with bounded derivatives, where one can take $k=2$, $m_x= 0$ and  $m_y= 0$. 
\begin{prop}\label{thm3.3}(version of \cite[Theorem 3.3]{crisan2022poisson})

\begin{enumerate}[label=\roman*)]
    \item \label{propfirstversion}
    Suppose the process \eqref{eqn:slowgeneral thm}-\eqref{eqn:fast general thm} satisfies the following assumptions
    \begin{enumerate}[label=\textnormal{[C\arabic*]},ref={[C\arabic*]}]
			\item\label{assgenthm1} 
            The coefficient $b$ is in $C^{2,4}(\R^{d_{\theta}}\times \R^{d_x})$ and there exist $m, m'>0$ such that all the derivatives $\partial_{\theta}^{\gamma}\partial_x^{\tilde{\gamma}} b$ with $0\leq 2 |\gamma|_*+|\tilde{\gamma}|_{*}\leq 4$ grow at most like $1+|\theta|^m+|x|^{m'}$. Here $|\gamma|_*$ denotes the length of the multi-index $\gamma$, i.e. $|\gamma|_*=\sum \gamma^i$. 
            \item \label{assgenthm2} The coefficient $a$ is such that for every $\theta$ fixed,  $a(\theta, \cdot) \in C^{4+\nu}(\R^{d_x}) $ for some $0<\nu<1$. And there exists $m''>0$ such that all the derivatives $\partial_{\theta}^{\gamma}\partial_x^{\tilde{\gamma}} a$ with $0\leq 2 |\gamma|_*+|\tilde{\gamma}|_{*}\leq 4$ are bounded in $\theta$ and grow at most like $1+|x|^{m''}$ in the variable $x$. 
            \item \label{assgenthm3}  The diffusion coefficients $\tilde b$ and $\tilde a$ are {smooth, }bounded, and uniformly elliptic; by the latter fact we mean that there exist constants $\lambda_-, \lambda_+>0$ such that
            $$
            \lambda_-\leq \langle \tilde a\tilde a^T(\theta) \xi/|\xi|, \xi/|\xi|\rangle \leq \lambda_+ \, ,
            $$
            and an analogous bound holds for $\tilde b$ as well. 
            \item \label{assgenthm4}  There exist constants $r, C, \tilde{r}, \tilde C>0$ such that
            \begin{equation}
                \langle a(\theta,x), x\rangle \leq 
                -r |x|^2+C \, ,
            \end{equation}
             and, similarly, 
             \begin{equation}
                \langle b(\theta,x), \theta\rangle \leq 
                -r |\theta|^2+C \, .
            \end{equation}
        \item \label{assgenthm5}  {
        For all $k \in \{2,4\}$ and $\psi\in C^{0,k}(\R^{d_\theta}\times\R^{d_x})$ such that, for $1 \leq \lvert\tilde{\gamma}\rvert_{*} \leq k$, $|\partial_{x}^{\tilde{\gamma}}\psi(\theta, x)|$ grows at most like $1+|\theta|^{m_\theta} + |x|^{m_x}$, there exist $K, \kappa>0$ such that for all $\theta \in \R^{d_\theta}$ and $x \in \R^{d_x}$
        \begin{equation}\label{eqn:4der}	\lvert\partial_{x}^{\tilde{\gamma}}\mathcal{P}_{t}^{\theta}\psi^\theta(x) \rvert\leq K(1+|\theta|^{m_\theta} + |x|^{m_x})e^{-\kappa t}.
		\end{equation}
		Moreover there exist constants $\tilde{K},C>0$ such that for any $\psi\in C_b^2(\R^{d_\theta})$ we have
		\begin{equation}\label{eqn:2der}
			\sup_{1 \leq \lvert\gamma\rvert_* \leq 2}\lVert \partial^{\gamma}_{\theta} \bar{\cP}_{t}\psi \rVert_\infty  \leq \tilde{K}e^{-C t}\sum_{1 \leq \lvert\gamma\rvert_* \leq 2} \lVert \partial^{\gamma}_{\theta} \psi \rVert^2_\infty.
		\end{equation}}
    \end{enumerate} 
Then there exists a constant $C>0$, independent of $t$ but maybe dependent on $\theta, x$\footnote{This dependence on the initial conditions is studied in \cite{crisan2022poisson}, but is not the main focus of this paper.}, such that 
\begin{equation}\label{eqn:uitsimple}
    |\mathbf E f(\theta_t) - \mathbf E f(\bar\theta)| \leq \epsilon C (\|\nabla f\|_{\infty}+ \|\nabla^2 f\|_{\infty}), 
\end{equation}
for every $f \in C_b^2(\R^{d_{\theta}})$. 
\item \label{propsecondversion}
{
If we also assume that the diffusion coefficients $\tilde{a}$ and $\tilde{b}$ are constant, i.e. for all $\theta, x$ we have $\tilde{a}(\theta) = \tilde{a}$ and $\tilde{b}(\theta) = \tilde{b}$ for some $\tilde{a}, \tilde{b} > 0$, then \ref{assgenthm1}, \ref{assgenthm2}, and \ref{assgenthm5} can be relaxed to the following
    \begin{enumerate}[label=\textnormal{[$\tilde{C}$\arabic*]},ref={[$\tilde{C}$\arabic*]}]
			\item\label{assgenthm1_simple} 
            The coefficient $b$ is in $C^{2,2}(\R^{d_{\theta}}\times \R^{d_x})$, grows at most like $1+|\theta|+|x|$ and all the derivatives $\partial_{\theta}^{\gamma}\partial_x^{\tilde{\gamma}} b$ with $1\leq  |\gamma|_*+|\tilde{\gamma}|_{*}\leq 2$ are bounded.
            \item \label{assgenthm2_simple} The coefficient $a \in C^{2,2}(\R^{d_\theta}\times \R^{d_x})$ is bounded in $\theta$, grows at most like $1+|x|$ in the variable $x$, and all the derivatives $\partial_{\theta}^{\gamma}\partial_x^{\tilde{\gamma}} a$ with $1\leq |\gamma|_*+|\tilde{\gamma}|_{*}\leq 2$ are bounded.
            \setcounter{enumii}{4}
        \item
        \label{assgenthm5_simple}         
        There exist $K, \kappa>0$ such that for all $\psi\in C^{0,2}(\R^{d_\theta}\times\R^{d_x})$ such that, for all $1\leq |\tilde{\gamma}|_{*}\leq 2$, $|\partial_{x}^{\tilde{\gamma}}\psi(\theta, x)|\leq \Psi$ for some constant $\Psi > 0$, the following holds for all $\theta \in \R^{d_\theta}$ and $x \in \R^{d_x}$
        \begin{equation}\label{eqn:4dersimple}
			\sup_{1\leq |\tilde{\gamma}|_{*}\leq 2}\|\partial_{x}^{\tilde{\gamma}}\mathcal{P}_{t}^{\theta}\psi^\theta(x) \|_\infty\leq K\cdot\Psi \cdot e^{-\kappa t}.
		\end{equation}
		Moreover there exist constants $\tilde{K},C>0$ such that for any $\psi\in C_b^2(\R^{d_\theta})$ we have
		\begin{equation}\label{eqn:2dersimple}
			\sup_{1 \leq \lvert\gamma\rvert_* \leq 2}\lVert \partial^{\gamma}_{\theta} \bar{\cP}_{t}\psi \rVert^2_\infty  \leq \tilde{K}e^{-C t}\sum_{1 \leq \lvert\gamma\rvert_* \leq 2} \lVert \partial^{\gamma}_{\theta} \psi \rVert^2_\infty .
		\end{equation}
    \end{enumerate} 
    That is, let $\tilde{a}$ and $\tilde{b}$ be constant, and further assume \ref{assgenthm1_simple}, \ref{assgenthm2_simple}, \ref{assgenthm4} and \ref{assgenthm5_simple}. Then \eqref{eqn:uitsimple} holds.
}
\end{enumerate}
\end{prop}
\begin{proof}
    The proof of Proposition~\ref{thm3.3} can be found in the Supplementary Material \cite{suppMat}.
\end{proof}

\begin{note}\label{note333}
Before moving on to the proof of Proposition \ref{thm3.3}, some brief comments to help understand the role of each assumption. More in depth comments on this can be found throughout the paper \cite{crisan2022poisson}, and in particular in \cite[Note 2.1 and Note 3.4]{crisan2022poisson}. 
Conditions \ref{assgenthm1}-\ref{assgenthm3}  are standard growth and smoothness assumptions. Condition \ref{assgenthm4} implies the existence of a Lyapunov function for the frozen process \eqref{eqn:frozengeneralthem} and for the  slow process \eqref{eqn:slowgeneral thm}  (first and second condition in \ref{assgenthm4}, respectively). Assumption \ref{assgenthm5} requires exponential decay in time of the space-derivatives of the frozen (equation \eqref{eqn:4der}) and averaged semigroups (equation \ref{eqn:2der}). That is, it requires SES of the frozen and averaged semigroups. 
To explain in a simplified setting why SES is key to proving uniform in time convergence, let us consider two Markov semigroups, say $\mathcal T_t$ and $\overline{\mathcal {T}}_t$. With standard manipulations, the difference between any two Markov semigroups can be expressed in terms of the difference between their respective generators, say $\mathcal G$ and 
			$\bar{\mathcal G}$, as follows
			\begin{align*}
				(\overline{\mathcal {T}}_t \varphi)(z) -  (\mathcal{T}_t \varphi)(z) & =  \int_0^t ds \frac{d}{ds} \mathcal{T}_{t-s}\bar{\mathcal{T}}_s \varphi (z)  =  
				\int_0^t \!\!\!ds \,  \mathcal{T}_{t-s} (\bar{\mathcal G} - \mathcal G) \bar{\mathcal T}_s \varphi (z) \\
				&\leq \int_0^t ds \| \mathcal T_{t-s} (\bar{\mathcal G} - \mathcal G) \bar{\mathcal T}_s \varphi \|_{\infty}
				\leq \int_0^t ds \|  (\bar{\mathcal G} - \mathcal G) \bar{\mathcal T}_s \varphi \|_{\infty}
			\end{align*}
		If $\mathcal G$ and $\bar{\mathcal{G}}$ are differential operators then the latter difference involves derivatives of the semigroup $\bar{\mathcal T}_t$. If such derivatives decay exponentially fast in time, then the difference between such semigroups can be estimated by a constant (independent of time) rather than with exponential growth, which is what would happen by using Gronwall-type arguments. This line of reasoning, applied to the semigroups generated by the averaged dynamics and by the system \eqref{eqn:slowgeneral thm}-\eqref{eqn:fast general thm}, inspires our approach - though the precise proof does not exactly follow the above calculation and some further manipulations are required (to obtain the correct power of $\epsilon$ on the RHS). These further manipulations then involve also the derivatives of the frozen semigroup, see e.g. \cite[equation (97)]{crisan2022poisson}. Finally, a framework to prove uniform in time results has been recently put forward in \cite{schuh2024conditions}. This framework emphasizes that contractivity (either of the dynamics itself or of the approximating dynamics), plays a pivotal role when trying to prove uniform in time estimates. SES implies contractivity and this is another way of viewing the appearance of such a condition in this context. The proof of the SES estimates is the one where the convexity assumptions on the coefficients is used the most. 
\end{note}

We now state two lemmas, Lemma \ref{lemma:fastsemigroupderests} and Lemma \ref{lem:avgderest}, which establish \cite[Equation (45)]{crisan2022poisson} and \cite[Equation (46)]{crisan2022poisson} for our purposes. We leave the proofs of these for the Supplementary Material \cite{suppMat}.

\begin{lem}\label{lemma:fastsemigroupderests}

	Let Assumption \ref{ass:growthcoeffs} and \ref{ass:strongconvexfast} hold. Then the following holds
	
	\begin{equation}\label{eqn:derestaveraging}
	\left\| \nabla_x P_t^{\theta} g \right\|^2_\infty \leq e^{-\kappa t}  \left\| \nabla_x  g \right\|^2_\infty
	\end{equation}
	for all $t\geq 0$ and $g \in \tilde{\mathcal{C}}^1_b(\R^{d_x}\times \R^{d_\theta})$, where $\kappa>0$ is as in \eqref{eqn:fastStrongConvex}.
	Furthermore, there exist constants $c, C > 0$ such that
	\begin{equation}\label{eqn:secondDerEstaveraging}
	\left\| \nabla^2_x P_t^{\theta} g \right\|^2_\infty \leq Ce^{-ct}\left(\left\| \nabla^2_x  g \right\|^2_\infty + \left\| \nabla_x  g\right\|^2_\infty\right)
	\end{equation}
	for all $t\geq 0$ and $g \in \tilde{\mathcal{C}}^2_b(\R^{d_x}\times \R^{d_\theta})$. Here $\tilde{\mathcal{C}}^1_b(\R^{d_x}\times \R^{d_\theta})$ is the space of continuous real-valued functions on $\R^{d_x}\times \R^{d_\theta}$ with first order derivatives being continuous and bounded, and $\tilde{\mathcal{C}}^2_b(\R^{d_x}\times \R^{d_\theta})$ is similar with the second order derivatives also being continuous and bounded.
\end{lem}
\begin{proof}
    The proof of Lemma~\ref{lemma:fastsemigroupderests} can be found in the Supplementary Material \cite{suppMat}.
\end{proof}

\begin{lem}\label{lem:avgderest}
	Let Assumption \ref{ass:growthcoeffs}-\ref{ass:avgderest} hold. Then there exist constants $\tilde{c},\tilde{C} > 0$ such that
	
	\begin{equation}\label{eqn:sgderestavged}
	\left\|\nabla\bar{P}^\beta_t f\right\|_\infty^2 +  \left\|\nabla^2 \bar{P}^\beta_tf\right\|_\infty^2 \leq \tilde{C}e^{-\tilde{c}t}\left(\left\|\nabla f\right\|_\infty^2 +  \left\|\nabla^2 f\right\|_\infty^2\right)
	\end{equation}
	for all $t\geq 0$, $\beta >0$, and $f \in \mathcal{C}^2_b(\R^{d_\theta})$.
\end{lem} 
\begin{proof}
    The proof of Lemma~\ref{lem:avgderest} can be found in the Supplementary Material \cite{suppMat}.
\end{proof}

\begin{rem}\label{rem:avgstronglyconvex}
	In proving Lemma \ref{lem:avgderest}, we will prove that under Assumption \ref{ass:growthcoeffs}-\ref{ass:avgderest}, the drift coefficient of $\bar \theta^{ \beta}_t$ from \eqref{eq:theta_bar_sde} is strongly convex. See the proof of Lemma \ref{lem:avgderest} in the Supplementary Material \cite{suppMat}.
\end{rem}

\begin{note}\label{note:polygrowth}
	Assumption \ref{ass:growthcoeffs} may be relaxed to allow for unbounded second and third derivative, i.e. the following: there exist $C,m,m'>0$ such that $$|\nabla^2 U(\theta,x)| + |\nabla^3 U(\theta,x)| \leq C(1+|x|^m+|\theta|^{m'}),$$
	but there must exist $\alpha >0$ such that  $|\partial_{x_i}\partial_{x_j}\partial_{x_k}U(\theta,x) | \leq \alpha $.
	In this case, \eqref{eqn:slowstrongconvex1} becomes
	\begin{equation}
	\begin{aligned}
	&- \sum_{i,j=1}^n \partial_{\theta_i} \partial_{\theta_j} U(\theta,x) \xi_i\xi_j \leq -\left(\frac{d_\theta D^2}{\kappa}+\zeta_0(1+|\theta|^{m'})\right) \lvert\xi\rvert^2,
	\end{aligned}
	\end{equation}
	and the final result \eqref{eqn:uitAveraging} becomes

 \begin{equation}
	\left \vert
	\mathbb E f(\theta_t^{\epsilon,\beta,  \theta, x}) - \mathbb E f(\bar{ \theta}^{\beta,\theta}_t)\right \vert \leq \epsilon C\left(\left\|\nabla f\right\|_\infty^2 +  \left\|\nabla^2 f\right\|_\infty^2\right)(1+|x|^{3m}+|\theta|^{4m'}).  
	\end{equation}
	The complication here is that one needs \eqref{eqn:derestaveraging} and \eqref{eqn:secondDerEstaveraging} to hold for $g$ with polynomial derivatives, rather than bounded derivatives (in particular, one needs them to hold for $g$ with the same growth as $|\nabla^2 U(\theta,x)|$ and $|\nabla^3 U(\theta,x)|$). This is not a problem, however, since one can obtain these in the polynomial setting via the arguments in the proof of \cite[Proposition 4.5]{crisan2022poisson}. {The calculation is the same, but significantly lengthier. We present the bounded second and third derivative case here for ease of notation and exposition.}
\end{note}



\section{Numerical examples}\label{sec:numex}

In the following, we compare the SFLA scheme (see Algorithm \ref{alg:ais}) to existing methods, namely to SOUL (see Section~\ref{sec:comparison_SOUL}, also \cite{de2021efficient}) and PGD (see \cite{kuntz2023particle}). More specifically, we implement Algorithm \ref{alg:ais} initialised at $\left(\theta^{\epsilon, \beta}_{0}, X^{\epsilon, \beta}_{0}\right) = (\theta, x)$.

We compare Algorithm \ref{alg:ais} with SOUL and PGD for two different reasons: SOUL because we believe it is comparable in concept to Algorithm \ref{alg:ais}, see Note \ref{note:cont}, and PGD as an example of another recently proposed algorithm which addresses the MMLE problem. We implement SOUL with no burn-in period (i.e. $\tilde{M}=M$ in \eqref{eqn:soul_latent_update}-\eqref{eqn:soul_theta_update}), and an initialisation of $\left( \theta_0, X^{(0)}_{0}\right) = (\theta, x)$. Similarly, our implementation of PGD has no burn in and is initialised at $X^{i,N}_0 = x$ for all $1\leq i \leq N$ and $\theta^N_0 = \theta$.

For ease of comparison between the algorithms, we introduce two parameters, $N$ and $\gamma$, which we vary. We use stepsize $\gamma$ with $N$ particles for our implementation of PGD. For our implementation of SOUL, we set $ \delta = \gamma$, and $\tilde{M}=M=N$. For our implementation of SFLA (see Algorithm \ref{alg:ais}), we set $\epsilon = 1/N$ and use a stepsize of $\delta = \gamma/N$. Before setting up the main example for which we run the experiments, we make the following note.
 \begin{rem}\label{note:cont} Two remarks are in order:
\begin{itemize}
\item We point out the similarity between Algorithm \ref{alg:ais} and diffusion based MMLE schemes such as SOUL (see \eqref{eqn:soul_latent_update}-\eqref{eqn:soul_theta_update}). Substantially, SOUL runs the FrP for a number of iterations (so, in averaging terms, can be seen as converging for each $k$ fixed, as $M \rightarrow\infty$ to the invariant measure $p_{\theta_{k-1}}(x|y)$ of the FrP), then averages over those iterations to update the SP, and repeats. It can be thought of as embodying the use of two different stepsizes for the parameter and latent variable dynamics, and is similar, at least in spirit, to the multiscale system from Algorithm \ref{alg:ais}. More on this comparison can be seen in Section \ref{sec:comparison_SOUL}.
\item Algorithm \ref{alg:ais} is an Euler discretisation of \eqref{eq:theta_update}-\eqref{latent_update}. Note that \eqref{eq:theta_update}-\eqref{latent_update} can be viewed as an SDE with stiff coefficients; it is well known \cite[Chapter 9]{weinan2011principles} that for such SDEs using an Euler discretisation is certainly not an optimal strategy, and indeed there are many methods including splitting and implicit schemes for discretising multiscale SDEs, for a review of these methods see \cite{weinan2011principles} and references therein. Here we started with implementing Euler just for simplicity. Further work is ongoing to produce similar theoretical results for the discretisation of multiscale systems.
\end{itemize}
 \end{rem}
 
\subsection{Bayesian logistic regression}\label{sec:blrExample}
For our numerical experiments, we consider the Bayesian logistic regression set-up described in \cite[Section 4.1]{de2021efficient}. We do not outline the entire setup here, but write the relevant potential, namely: 

\begin{align}
U(\theta, x)&= \frac{d_x}{2}\log(2\pi \sigma^2) + \sum^{d_y}_{i=1}\left( y_i \log(s(v_i^T x)) +(1-y_i)\log((s(-v_i^Tx)))\right) - \frac{|x-\theta \mathbf{1}_{d_x}|^2}{2\sigma^2},
\end{align}
with all the relevant notation in \cite[Section 4.1]{de2021efficient}.

\begin{figure}[t]
\begin{center}
\includegraphics[width=1\textwidth]{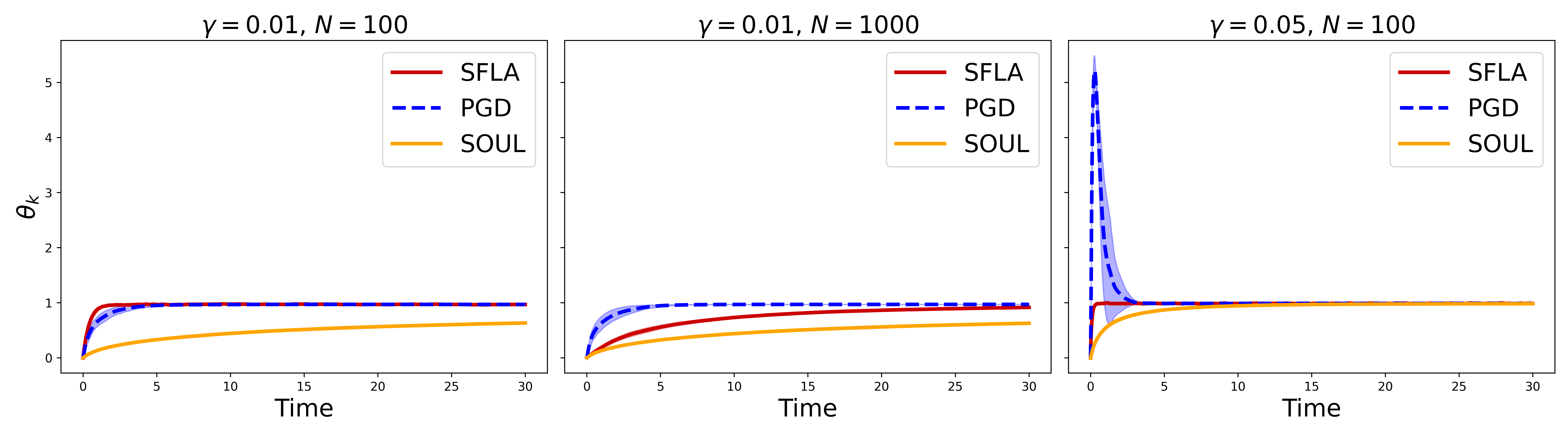}
 \caption{The performance of the algorithms PGD, SOUL and our multiscale \emph{slow fast} method (SFLA) can be seen below. This is in the context of Section \ref{sec:blrExample}. We plot the mean chain $\theta_t$ for each algorithm ($10$ runs were done, with $\pm 1$ standard deviation also plotted), with differing $\gamma$ and $N$. Each algorithm was run for 30 seconds, with initialisation $(\theta,x) = (0,0)$, and $\beta = 10^{4}$. The instability of PGD for $\gamma = 0.05$ and $N=100$ may be surprising but this is due to the lack of a burn-in period.}
 \label{fig:blrTimes}
 \end{center}
\end{figure}

We now numerically compare the method outlined in Algorithm \ref{alg:ais} to both SOUL (see \cite{de2021efficient}) and PGD (see \cite{kuntz2023particle}). All three algorithms we use here can be seen as approximations of \eqref{eq:theta_bar_sde}. We point out that these are very preliminary results, which do not by themselves justify the use of the discretisation in Algorithm \ref{alg:ais} over alternatives such as SOUL and PGD. For this, more thorough experiments in a wider variety of contexts (for example with non-convex potential $U$, or experiments that utilise PGD's parallelisation capabilities) would be needed. Instead, we use this section to simply compare Algorithm \ref{alg:ais} to existing methods which address the MMLE problem, in light of the observations made in Note \ref{note:cont} about the similarity between the three that we present here.

As mentioned at the start of Section \ref{sec:numex}, for our implementation of Algorithm \ref{alg:ais} we pick $\epsilon = 1/N$. This is so that all algorithms have the same computational complexity, that is, they all make $\mathcal{O}(KN)$ gradient computations, where $K$ is the number of iterations. Though this is the case, PGD allows for mitigating the factor of $N$ by vectorising over the particle cloud; SOUL and our slow-fast algorithm do not, since time cannot be parallelised over. Hence, we plot the parameter estimate against the time run, in Figure \ref{fig:blrTimes}, as well as the computations, in Figure \ref{fig:blrComputation}. We can see in Figure \ref{fig:blrTimes} that PGD can converge with a smaller stepsize than SOUL and our slow fast algorithm. However, SOUL and the slow fast algorithm seem slightly more stable than PGD under larger stepsize. When one plots the parameter estimate versus gradient computation, as in Figure \ref{fig:blrComputation}, one can see the similarity between SOUL and Algorithm \ref{alg:ais}. This is precisely because they are both approximating \eqref{eq:theta_bar_sde} from a time perspective, rather than the particle perspective. This being said, PGD also looks to perform very similarly when plotted against computations.

\begin{figure}[t]
\begin{center}
\includegraphics[width=1\textwidth]{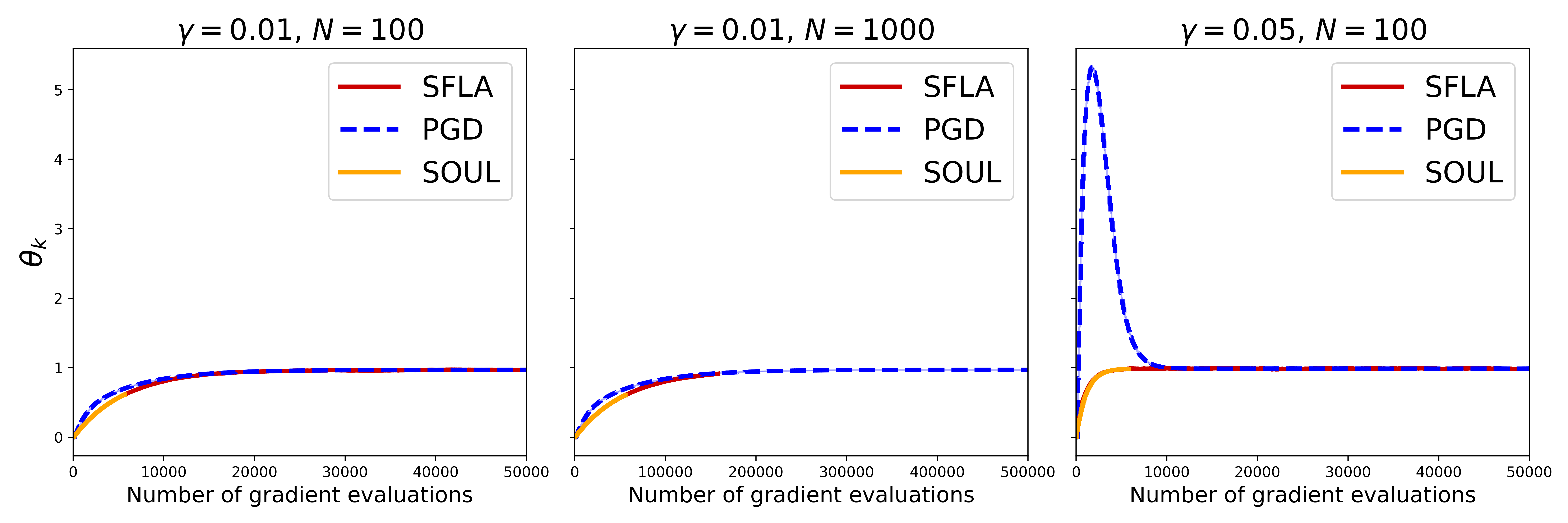}
 \caption{This is the same data as that of Figure \ref{fig:blrTimes}, except we plot mean $\theta_k$ (again, $10$ runs were done, with $\pm 1$ standard deviation also plotted) against the number of computations (which we define simply as the number of times a gradient is evaluated), rather than the run time. This is useful to see how the ability for PGD to be parallelised is more beneficial for larger $N$, as one would expect. Again, here the initialisation is $(\theta,x) = (0,0)$, $\beta = 10^{4}$ and each algorithm was run for 30 seconds; hence the early stopping of some lines. We also see that the vast majority of variability across runs was in the time taken (i.e. is seen in Figure \ref{fig:blrTimes}), rather than the value of $\theta_k$ after a certain number of gradient evaluations.}
\label{fig:blrComputation}
\end{center}
\end{figure}
\subsection{Bayesian neural network}
We now consider the setting of \cite{Yao2020StackingFN} to classify images using the MNIST dataset, denoted hereafter as $\mathcal{Y}$. Another implementation of this example can be seen in \cite[Section 3.2 and Section E.3]{kuntz2023particle}. We subsample $1000$ datapoints of digits labelled as $4$ and $9$ and split the resulting dataset into training and test sets, $\mathcal{Y}_{train}$ and $\mathcal{Y}_{test}$ with proportion $80\%$ and $20\%$ respectively.
\begin{figure}[t]
\centering
\includegraphics[width=\textwidth]{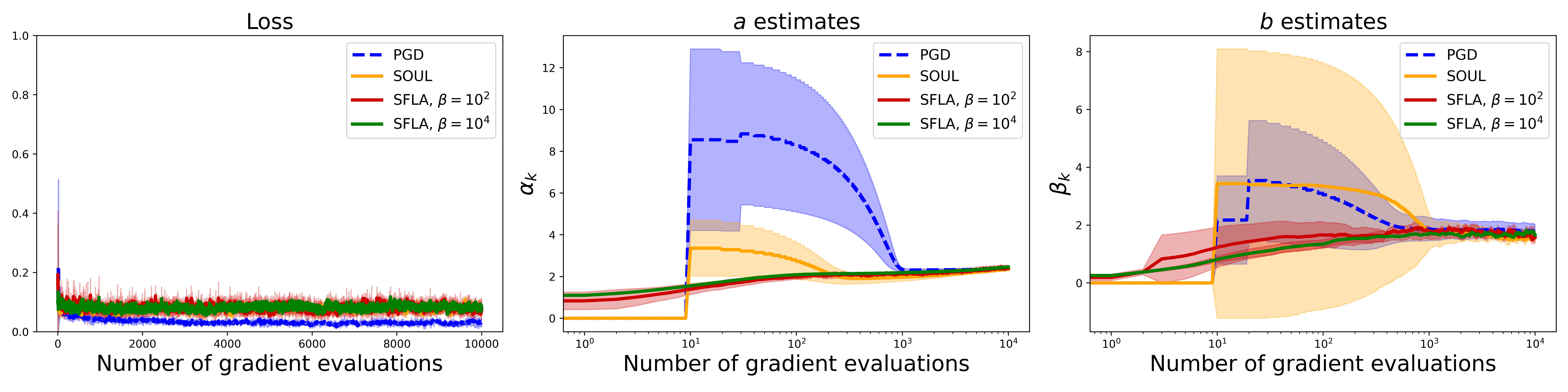}
\caption{Bayesian neural network example. This was produced with $\gamma=0.1$ and $N=10$, with $K=1000$ steps for the PGD and SOUL algorithms. We plot against computations rather than iterations in the same way as in Figure \ref{fig:blrComputation}, so that each algorithm can be compared fairly on the same axis. Here we ran each algorithm $5$ times and plot $\pm 1$ standard deviation along with the mean. For SFLA, we plot the results for two different $\beta$ values: $10^2$ and $10^4$.}
\label{fig:bnn}
\end{figure}

In the following we use the same notation as \cite[Section 3.2 and Section E.3]{kuntz2023particle}, except using $a$ and $b$ instead of $\alpha$ and $\beta$. We gather the likelihood of the model here for the reader's convenience. A two-layer neural network with tanh activation functions and softmax output layer is used. The input layer has 40 nodes and 784 inputs, and the output layer has two nodes. The weights of the neural network ($x\coloneqq (w,v) \in \R^{D_w} \times \R^{D_v}$ with $D_w = 40\times 784$ and $D_v = 2\times 40$) are the latent variables, and a zero-mean Gaussian prior is put on the weights with variances $e^{2a}$ and $e^{2b}$ respectively. Hence, the likelihood of the labels $l$ given features $f$ and network weights $x$ is given by
\begin{equation}\label{eqn:likelihoodBNN}
    p(l|f,x) \propto \text{exp} \left( \sum^{40}_{j=1} v_{lj} \tanh \left(\sum^{784}_{i=1} w_{ji}f_i\right) \right).
\end{equation}
 Our slow variable here is $\theta \coloneqq (a, b)$, and by \eqref{eqn:likelihoodBNN} we have that the model density is
\begin{equation}
    p_\theta (x|\mathcal{Y}_{train}) = \mathcal{N}(w;\mathbf{0}_{D_w}, e^{2a}I_{D_w})\mathcal{N}(v;\mathbf{0}_{D_v}, e^{2b}I_{D_v}) \prod_{(f,l) \in \mathcal{Y}_{train}} p(l |f,x),
\end{equation}
where $\mathbf{0}_{D_w}$ is the zero vector of dimension $D_w$ and $I_{D_w}$ is the identity matrix of size $D_w$.
The numerics we perform for this example are similar to Section \ref{sec:blrExample}. We implement Algorithm \ref{alg:ais}, SOUL (see \eqref{eqn:soul_latent_update}-\eqref{eqn:soul_theta_update}) and PGD (see \eqref{eq:ContIPS_theta_kuntz}-\eqref{eq:ContIPS_x_kuntz}). We use stepsize $\gamma = 0.1$ with $10$ particles for our implementation of PGD. For our implementation of SOUL, we set $ \delta = \gamma = 0.1$, and $\tilde{M}=M=N =10$. For our implementation of SFLA (see Algorithm \ref{alg:ais}), we set $\epsilon = 1/10$ and use a stepsize of $\delta = \gamma/N = 0.01$.

We note that this example does not satisfy the conditions needed for Theorem \ref{thm:main_thm}. In particular, the corresponding $U(\theta, x) \coloneqq -\log  p_\theta (x|\mathcal{Y}_{train})$ is well known to be, in most cases, multimodal. Hence, we include it simply as indicative of possible generalisations to Theorem \ref{thm:main_thm}. The results are shown in Figure \ref{fig:bnn}, with both the classification error and each component of the parameter update plotted for each iteration. We see similar behaviour between SOUL and the slow fast algorithm, with slightly better but still comparable performance from PGD. {This example also gives some evidence that the tuning of $\beta$ may be important in the case where the potential $U$ is not convex. In Figure \ref{fig:bnn} for the estimate of $b$, one can see that although the estimates with $\beta = 10^2$ have higher variance than that of $\beta = 10^4$, the runs with $\beta = 10^2$ seem to indicate quicker convergence. This is, however, not definitive and simply indicates that $\beta$ plays a role beyond simple concentration around the maximiser. We also note that both SOUL and SFLA exhibit higher variance in the \emph{left} of Figure \ref{fig:bnn}. This corresponds to the loss of the Bayesian neural network. Hence, it is based on the weights which are governed by the fast dynamics. The \emph{middle} and \emph{right} hand graphs are of the slow variable, and the higher variance in the loss does not seem to translate to higher variance of the parameters $a$ and $b$.}

\section{Conclusions}
In this paper, we have introduced slow-fast SDEs which enabled us to see various parameter estimation methods from a multiscale perspective. Using the multiscale perspective, we also investigated a two time-scale system, termed Slow-Fast Langevin Algorithm (SFLA), for parameter estimation in statistical models. Leveraging recent results in the stochastic averaging literature, we provided a detailed theoretical analysis and discussion as well as numerical comparisons to relevant algorithms from computational statistics literature.

Our perspective opens up many possibilities for merging results from applied probability (multiscale analysis) and computational statistics, as many prominent methods in computational statistics (like the EM algorithm) use a combination of sampling and optimisation methods which can be suitably rewritten as multiscale systems. We expect that the perspective we provide here will be the foundation of many such theoretical results as well as practical new developments motivated by the multiscale perspective.

\section*{Acknowledgements}
The authors would like to thank the Isaac Newton Institute for Mathematical Sciences, Cambridge, for support and hospitality during the programme \textit{The mathematical and statistical foundation of future data-driven engineering}, where work on this paper was undertaken. This work was supported by EPSRC grant EP/R014604/1. I. Souttar was partially supported by a MAC-MIGS CDT Scholarship under EPSRC grant EP/S023291/1. Michela Ottobre gratefully acknowledges support from the EPSRC grant EP/Z534225/1 . 

\bibliography{slow_fast_arxivsubmission_abbrv.bib}

@article{akyildiz2023interacting,
  title={Interacting particle {L}angevin algorithm for maximum marginal likelihood estimation},
  author={Akyildiz, O Deniz and Crucinio, Francesca Romana and Girolami, Mark and Johnston, Tim and Sabanis, Sotirios},
  journal={ESAIM: Probab. Stat.},
  volume={29},
  pages={243--280},
  year={2025},
  publisher={EDP Sciences}
}

@article{schuh2024conditions,
  title={Conditions for uniform in time convergence: applications to averaging, numerical discretisations and mean-field systems},
  author={Schuh, Katharina and Souttar, Iain},
  journal={arXiv preprint arXiv:2412.05239},
  year={2024}
}

@article{chaudhari2018deep,
  title={Deep relaxation: partial differential equations for optimizing deep neural networks},
  author={Chaudhari, Pratik and Oberman, Adam and Osher, Stanley and Soatto, Stefano and Carlier, Guillaume},
  journal={Res. Math. Sci.},
  volume={5},
  pages={1--30},
  year={2018},
  publisher={Springer}
}

@inproceedings{marion2025implicit,
  title={Implicit diffusion: efficient optimization through stochastic sampling},
  author={Marion, Pierre and Korba, Anna and Bartlett, Peter and Blondel, Mathieu and De Bortoli, Valentin and Doucet, Arnaud and Llinares-L{\'o}pez, Felipe and Paquette, Courtney and Berthet, Quentin},
  booktitle={The 28th International Conference on Artificial Intelligence and Statistics},
  year={2025}
}

@article{breiten2021stochastic,
  title={Stochastic gradient descent and fast relaxation to thermodynamic equilibrium: a stochastic control approach},
  author={Breiten, Tobias and Hartmann, Carsten and Neureither, Lara and Sharma, Upanshu},
  journal={J. Math. Phys.},
  volume={62},
  number={12},
  year={2021},
  publisher={AIP Publishing}
}

@article{barre2021fast,
  title={Fast non-mean-field networks: uniform in time averaging},
  author={Barr{\'e}, Julien and Dobson, Paul and Ottobre, Michela and Zatorska, Ewelina},
  journal={SIAM J. Math. Anal.},
  volume={53},
  number={1},
  pages={937--972},
  year={2021},
  publisher={SIAM}
}

@article{dobson2023infinite,
  title={Infinite dimensional piecewise deterministic {M}arkov processes},
  author={Dobson, Paul and Bierkens, Joris},
  journal={Stochastic Process. Appl.},
  volume={165},
  pages={337--396},
  year={2023},
  publisher={Elsevier}
}

@article{crisan2021uniform,
  title={Uniform in time estimates for the weak error of the {E}uler method for {SDE}s and a pathwise approach to derivative estimates for diffusion semigroups},
  author={Crisan, Dan and Dobson, Paul and Ottobre, Michela},
  journal={Trans. Amer. Math. Soc.},
  volume={374},
  number={5},
  pages={3289--3330},
  year={2021}
}

@inproceedings{kantas2019sharp,
  title     = {The Sharp, the Flat and the Shallow: Can Weakly Interacting Agents Learn to Escape Bad Minima?},
  author    = {Kantas, Nikolas and Parpas, Panos and Pavliotis, Grigorios A.},
  booktitle = {Proceedings of the ICML 2019 Workshop on AI in Finance: Applications and Infrastructure for Multi-Agent Learning},
  year      = {2019},
  address   = {Long Beach, CA, USA},
  month     = jun,
  note      = {Workshop held on June 14, 2019}
}

@article{stoltz2018longtime,
  title={Longtime convergence of the temperature-accelerated molecular dynamics method},
  author={Stoltz, Gabriel and Vanden-Eijnden, Eric},
  journal={Nonlinearity},
  volume={31},
  number={8},
  pages={3748},
  year={2018},
  publisher={IOP Publishing}
}

@book{ethierKurtz,
	title={Markov Processes: Characterization and Convergence},
	author={Ethier, Stewart N and Kurtz, Thomas G},
	year={2009},
	publisher={John Wiley \& Sons}
}

@article{rockner2020diffusion,
  title={Diffusion approximation for fully coupled stochastic differential equations},
  author={R{\"o}ckner, Michael and Xie, Longjie},
  journal={Ann. Probab.},
  volume={49},
  number={3},
  pages={1205--1236},
  year={2021},
  month={May}
}

@book{LiuRockner,
	title={Stochastic Partial Differential Equations: An Introduction},
	author={Liu, Wei and R{\"o}ckner, Michael},
	year={2015},
	publisher={Springer}
}

@article{GlynnMeyn,
	author = {Peter W. Glynn and Sean P. Meyn},
	title = {{A Liapounov bound for solutions of the Poisson equation}},
	volume = {24},
	journal = {Ann. Probab.},
	number = {2},
	publisher = {Institute of Mathematical Statistics},
	pages = {916--931},
	year = {1996}
}

@article{pardoux2001,
	author = "Pardoux, \'Etienne and Veretennikov, Alexander Yu.",
	fjournal = "Annals of Probability",
	journal = "Ann. Probab.",
	month = "07",
	number = "3",
	pages = "1061--1085",
	publisher = "The Institute of Mathematical Statistics",
	title = "On the {P}oisson equation and diffusion approximation {I}",
	volume = "29",
	year = "2001"
}

@article{pardoux2003,
	author = "Pardoux, \'Etienne and Veretennikov, Alexander Yu.",
	fjournal = "Annals of Probability",
	journal = "Ann. Probab.",
	month = "07",
	number = "3",
	pages = "1166--1192",
	publisher = "The Institute of Mathematical Statistics",
	title = "On {P}oisson equation and diffusion approximation {2}",
	url = "https://doi.org/10.1214/aop/1055425774",
	volume = "31",
	year = "2003"
}

@article{mattingly2002ergodicity,
  title={Ergodicity for {SDE}s and approximations: locally {L}ipschitz vector fields and degenerate noise},
  author={Mattingly, Jonathan C and Stuart, Andrew M and Higham, Desmond J},
  journal={Stochastic Process. Appl.},
  volume={101},
  number={2},
  pages={185--232},
  year={2002},
  publisher={Elsevier}
}

@book{weinan2011principles,
  title={Principles of Multiscale Modeling},
  author={Weinan, E},
  year={2011},
  publisher={Cambridge University Press}
}

@article{crisan2022poisson,
title = {{P}oisson equations with locally-{L}ipschitz coefficients and uniform in time averaging for stochastic differential equations via strong exponential stability},
author = "Dan Crisan and Paul Dobson and Ben Goddard and Michela Ottobre and Iain Souttar",
year = "2024",
journal = "Ann. Inst. Henri Poincar{\'e} Probab. Stat.",
issn = "0246-0203",
publisher = "Institute of Mathematical Statistics",
}

@article{angeli2023uniform,
  title={Uniform in time convergence of numerical schemes for stochastic differential equations via strong exponential stability: {E}uler methods, split-step and tamed schemes},
  author={Angeli, Letizia and Crisan, Dan and Ottobre, Michela},
  journal={ESAIM Math. Model. Numer. Anal.},
  year={2025},
  note={to appear}
}

@article{Yao2020StackingFN,
  title={Stacking for non-mixing {B}ayesian computations: the curse and blessing of multimodal posteriors},
  author={Yuling Yao and Aki Vehtari and Andrew Gelman},
  journal={J. Mach. Learn. Res.},
  year={2020},
  volume={23},
  pages={79:1--79:45},
  url={https://api.semanticscholar.org/CorpusID:219966748}
}

@book{pavliotis2008multiscale,
	title={Multiscale Methods: Averaging and Homogenization},
	author={Pavliotis, Grigoris and Stuart, Andrew M.},
	year={2008},
	publisher={Springer Science \& Business Media}
}

@article{pavliotis2022derivative,
  title={Derivative-free {B}ayesian inversion using multiscale dynamics},
  author={Pavliotis, Grigorios A and Stuart, Andrew M and Vaes, Urbain},
  journal={SIAM J. Appl. Dyn. Syst.},
  volume={21},
  number={1},
  pages={284--326},
  year={2022},
  publisher={SIAM}
}

@article{suppMat,
    author = {Akyildiz, O Deniz and Ottobre, Michela and Souttar, Iain},
    title = {Supplement to ``A multiscale perspective on maximum marginal
likelihood estimation''},
    year = {2026},
doi = {}
}

@inproceedings{kuntz2023particle,
  title={Particle algorithms for maximum likelihood training of latent variable models},
  author={Kuntz, Juan and Lim, Jen Ning and Johansen, Adam M},
  booktitle={International Conference on Artificial Intelligence and Statistics},
  pages={5134--5180},
  year={2023},
  organization={PMLR}
}

@article{whiteley2022discovering,
title = "Statistical exploration of the manifold hypothesis",
author = "Nick Whiteley and Annie Gray and Patrick Rubin-Delanchy",
year = "2025",
month = mar,
day = "12",
journal = "J. R. Stat. Soc. Ser. B Stat. Methodol.",
}

@article{zhou2018fenchel,
  title={On the {F}enchel duality between strong convexity and {L}ipschitz continuous gradient},
  author={Zhou, Xingyu},
  journal={arXiv preprint arXiv:1803.06573},
  year={2018}
}

@article{caprio2024error,
  title={Error bounds for particle gradient descent, and extensions of the log-{S}obolev and {T}alagrand inequalities},
  author={Caprio, Rocco and Kuntz, Juan and Power, Samuel and Johansen, Adam M},
  journal={J. Mach. Learn. Res.},
  volume={26},
  number={103},
  pages={1--38},
  year={2025}
}

@article{de2021efficient,
  title={Efficient stochastic optimisation by unadjusted {L}angevin {M}onte {C}arlo},
  author={De Bortoli, Valentin and Durmus, Alain and Pereyra, Marcelo and Vidal, Ana F},
  journal={Stat. Comput.},
  volume={31},
  number={3},
  pages={1--18},
  year={2021},
  publisher={Springer}
}

@article{gardner2002brunn,
  title={The {B}runn-{M}inkowski inequality},
  author={Gardner, Richard},
  journal={Bull. Amer. Math. Soc.},
  volume={39},
  number={3},
  pages={355--405},
  year={2002}
}

@article{hwang1980laplace,
  title={{L}aplace's method revisited: weak convergence of probability measures},
  author={Hwang, Chii-Ruey},
  journal={Ann. Probab.},
  pages={1177--1182},
  year={1980},
  volume={8},
  number={6},
  publisher={JSTOR}
}

@article{pavliotis2006introduction,
  title={An introduction to multiscale methods},
  author={Pavliotis, GA and Stuart, AM},
  journal={Lecture Notes},
  year={2006}
}

@article{dempster1977maximum,
  title={Maximum likelihood from incomplete data via the {EM} algorithm},
  author={Dempster, Arthur P and Laird, Nan M and Rubin, Donald B},
  journal={J. R. Stat. Soc. Ser. B Stat. Methodol.},
  volume={39},
  pages={2--38},
  year={1977},
  publisher={JSTOR}
}

@article{delyon1999convergence,
  title={Convergence of a stochastic approximation version of the {EM} algorithm},
  author={Delyon, Bernard and Lavielle, Marc and Moulines, Eric},
  journal={Ann. Statist.},
  volume={27},
  number={1},
  pages={94--128},
  year={1999},
  month={Feb},
  publisher={Institute of Mathematical Statistics}
}

@article{wei1990monte,
  title={A {Monte Carlo implementation of the EM} algorithm and the poor man's data augmentation algorithms},
  author={Wei, Greg CG and Tanner, Martin A},
  journal={J. Amer. Statist. Assoc.},
  volume={85},
  number={411},
  pages={699--704},
  year={1990},
  publisher={Taylor \& Francis}
}

@article{wu1983convergence,
  title={On the convergence properties of the {EM} algorithm},
  author={Wu, C. F. Jeff},
  journal={Ann. Statist.},
  volume={11},
  number={1},
  pages={95--103},
  year={1983},
  month={Mar},
  publisher={Institute of Mathematical Statistics}
}

@inproceedings{dalalyan2017further,
  title={Further and stronger analogy between sampling and optimization: {Langevin Monte Carlo} and gradient descent},
  author={Dalalyan, Arnak S},
  booktitle={Conference on Learning Theory},
  pages={678--689},
  year={2017},
  organization={PMLR}
}

@article{roberts1996exponential,
  title={Exponential convergence of {Langevin} distributions and their discrete approximations},
  author={Roberts, Gareth O and Tweedie, Richard L},
  journal={Bernoulli},
  volume={2},
  number={4},
  pages={341--363},
  year={1996},
  month={Dec}
}

@article{dalalyan2017theoretical,
  title={Theoretical guarantees for approximate sampling from smooth and log-concave densities},
  author={Dalalyan, Arnak S},
  journal={J. R. Stat. Soc. Ser. B Stat. Methodol.},
  volume={79},
  number={3},
  pages={651--676},
  year={2017},
  month={Jun}
}

@article{dalalyan2019user,
  title={User-friendly guarantees for the {L}angevin {M}onte {C}arlo with inaccurate gradient},
  author={Dalalyan, Arnak S and Karagulyan, Avetik},
  journal={Stochastic Process. Appl.},
  volume={129},
  number={12},
  pages={5278--5311},
  year={2019},
  publisher={Elsevier}
}

@article{durmus2019high,
  title={High-dimensional {B}ayesian inference via the unadjusted {L}angevin algorithm},
  author={Durmus, Alain and Moulines, \'Eric},
  journal={Bernoulli},
  volume={25},
  number={4A},
  pages={2854--2882},
  year={2019}
}

@article{carlin2000empirical,
  title={Empirical {B}ayes: past, present and future},
  author={Carlin, Bradley P and Louis, Thomas A},
  journal={J. Amer. Statist. Assoc.},
  volume={95},
  number={452},
  pages={1286--1289},
  year={2000},
  publisher={Taylor \& Francis}
}

@article{caffo2005ascent,
  title={Ascent-based {M}onte {C}arlo expectation--maximization},
  author={Caffo, Brian S and Jank, Wolfgang and Jones, Galin L},
  journal={J. R. Stat. Soc. Ser. B Stat. Methodol.},
  volume={67},
  number={2},
  pages={235--251},
  year={2005},
  publisher={Wiley Online Library}
}

@article{fort2011convergence,
 ISSN = {00905364, 21688966},
 author = {G. Fort and E. Moulines and P. Priouret},
 journal = {Ann. Statist.},
 number = {6},
 pages = {3262--3289},
 publisher = {Institute of Mathematical Statistics},
 title = {Convergence of adaptive and interacting {M}arkov chain {M}onte {C}arlo algorithms},
 urldate = {2023-03-21},
 volume = {39},
 year = {2011}
}

@article{atchade2017perturbed,
  title={On perturbed proximal gradient algorithms},
  author={Atchad{\'e}, Yves F and Fort, Gersende and Moulines, Eric},
  journal={J. Mach. Learn. Res.},
  volume={18},
  number={1},
  pages={310--342},
  year={2017},
  publisher={JMLR. org}
}

@article{fort2003convergence,
  title={Convergence of the {M}onte {C}arlo expectation maximization for curved exponential families},
  author={Fort, Gersende and Moulines, Eric},
  journal={Ann. Statist.},
  volume={31},
  number={4},
  pages={1220--1259},
  year={2003},
  publisher={Institute of Mathematical Statistics}
}

@article{meng1993maximum,
  title={Maximum likelihood estimation via the {ECM} algorithm: a general framework},
  author={Meng, Xiao-Li and Rubin, Donald B},
  journal={Biometrika},
  volume={80},
  number={2},
  pages={267--278},
  year={1993},
  publisher={Oxford University Press}
}

@article{liu1994ecme,
  title={The {ECME} algorithm: a simple extension of {EM} and {ECM} with faster monotone convergence},
  author={Liu, Chuanhai and Rubin, Donald B},
  journal={Biometrika},
  volume={81},
  number={4},
  pages={633--648},
  year={1994},
  publisher={Oxford University Press}
}

@article{lange1995gradient,
  title={A gradient algorithm locally equivalent to the {EM} algorithm},
  author={Lange, Kenneth},
  journal={J. R. Stat. Soc. Ser. B Methodol.},
  volume={57},
  number={2},
  pages={425--437},
  year={1995},
  publisher={Wiley Online Library}
}

@book{bishop2006pattern,
  title={Pattern Recognition and Machine Learning},
  author={Bishop, Christopher M},
  volume={4},
  number={4},
  year={2006},
  publisher={Springer}
}

@article{hoff2002latent,
  title={Latent space approaches to social network analysis},
  author={Hoff, Peter D and Raftery, Adrian E and Handcock, Mark S},
  journal={J. Amer. Statist. Assoc.},
  volume={97},
  number={460},
  pages={1090--1098},
  year={2002},
  publisher={Taylor \& Francis}
}

@article{sherman1999conditions,
  title={Conditions for convergence of {Monte Carlo EM} sequences with an application to product diffusion modeling},
  author={Sherman, Robert P and Ho, Yu-Yun K and Dalal, Siddhartha R},
  journal={Econom. J.},
  volume={2},
  number={2},
  pages={248--267},
  year={1999},
  publisher={Oxford University Press Oxford, UK}
}

@article{chan1995monte,
  title={Monte {Carlo EM} estimation for time series models involving counts},
  author={Chan, KS and Ledolter, Johannes},
  journal={J. Amer. Statist. Assoc.},
  volume={90},
  number={429},
  pages={242--252},
  year={1995},
  publisher={Taylor \& Francis}
}

@article{smaragdis2006probabilistic,
  title={A probabilistic latent variable model for acoustic modeling},
  author={Smaragdis, Paris and Raj, Bhiksha and Shashanka, Madhusudana},
  journal={Advances in Models for Acoustic Processing Workshop, NIPS},
  volume={148},
  pages={8--1},
  year={2006}
}

@article{celeux1992stochastic,
  title={A stochastic approximation type {EM} algorithm for the mixture problem},
  author={Celeux, Gilles and Diebolt, Jean},
  journal={Stochastics},
  volume={41},
  number={1-2},
  pages={119--134},
  year={1992},
  publisher={Taylor \& Francis}
}

@article{cappe1999simulation,
  title={Simulation-based methods for blind maximum-likelihood filter identification},
  author={Capp{\'e}, Olivier and Doucet, Arnaud and Lavielle, Marc and Moulines, Eric},
  journal={Signal Process.},
  volume={73},
  number={1-2},
  pages={3--25},
  year={1999},
  publisher={Elsevier}
}

@article{booth1999maximizing,
  title={Maximizing generalized linear mixed model likelihoods with an automated {Monte Carlo EM} algorithm},
  author={Booth, James G and Hobert, James P},
  journal={J. R. Stat. Soc. Ser. B Stat. Methodol.},
  volume={61},
  number={1},
  pages={265--285},
  year={1999},
  publisher={Wiley Online Library}
}

@InProceedings{pmlr-v65-raginsky17a,
  title = 	 {Non-convex learning via Stochastic Gradient {Langevin} Dynamics: a nonasymptotic analysis},
  author = 	 {Raginsky, Maxim and Rakhlin, Alexander and Telgarsky, Matus},
  booktitle = 	 {Proceedings of the 2017 Conference on Learning Theory},
  pages = 	 {1674--1703},
  year = 	 {2017},
  editor = 	 {Kale, Satyen and Shamir, Ohad},
  volume = 	 {65},
  series = 	 {Proceedings of Machine Learning Research},
  month = 	 {07--10 Jul},
  publisher =    {PMLR}
}

@article{celeux1985sem,
  title={The {SEM} algorithm: a probabilistic teacher algorithm derived from the {EM} algorithm for the mixture problem},
  author={Celeux, Gilles},
  journal={Comput. Stat. Q.},
  volume={2},
  pages={73--82},
  year={1985}
}

@incollection{diebolt1995stochastic,
  title={A stochastic {EM} algorithm for approximating the maximum likelihood estimate},
  author={Diebolt, J and Ip, E HS},
  year={1996},
  booktitle = {Markov Chain Monte Carlo in Practice},
  editor = {W. R. Gilks and S. T. Richardson and D. J. Spiegelhalter}
}

@article{kevrekidis2009equation,
  title={Equation-free multiscale computation: Algorithms and applications},
  author={Kevrekidis, Ioannis G and Samaey, Giovanni},
  journal={Ann. Rev. Phys. Chem.},
  volume={60},
  pages={321--344},
  year={2009},
  publisher={Annual Reviews}
}

@article{tamd_rare_event,
author = {Maragliano, Luca and Vanden-Eijnden, Eric},
year = {2006},
month = {05},
pages = {168-175},
title = {A temperature accelerated method for sampling free energy and determining reaction pathways in rare events simulations},
volume = {426},
journal = {Chem. Phys. Lett.}
}

@article{md_free_energy_abrams,
author = {Abrams, Jerry B. and Tuckerman, Mark E.},
title = {Efficient and direct generation of multidimensional free energy surfaces via adiabatic dynamics without coordinate transformations},
journal = {J. Phys. Chem. B},
volume = {112},
number = {49},
pages = {15742-15757},
year = {2008},
}

@article{tamd_free_energy_beta,
    author = {Maragliano, Luca and Vanden-Eijnden, Eric},
    title = {Single-sweep methods for free energy calculations},
    journal = {J. Chem. Phys.},
    volume = {128},
    number = {18},
    pages = {184110},
    year = {2008},
    month = {05}
}

\appendix

\section{Lemmas}
\subsection{Prekopa-Leindler inequality for $\mu$-strongly log-concave densities}\label{app:PLI}
The following lemma is well-known, see, e.g., \cite[Theorem~7.1]{gardner2002brunn}, we state it here for completeness.
\begin{lem}\label{lem:PLI} (Prekopa-Leindler Inequality) Let $h(x, y)$ be a jointly-log-concave density (i.e. $- \nabla \log h(z)$ satisfies Assumption~\ref{assmp:convexity} with $\mu = 0$ with $z = (x, y)$) which implies
\begin{align*}
h((1-\lambda) z + \lambda z') \geq h(z)^{1-\lambda} h(z')^\lambda.
\end{align*}
Then, $\tilde{h}(y) = \int h(x, y) \mathrm{d}x$ is also log-concave, i.e.,
\begin{align*}
\tilde{h}((1-\lambda) y + \lambda y') \geq \tilde{h}(y)^{1-\lambda} \tilde{h}(y')^{\lambda}.
\end{align*}
\end{lem}
The idea can be extended to $\mu$-strongly log-concave densities. Consider a $\mu$-strongly log-concave function $h(x,y)$, i.e. $h(x, y) \propto e^{-V(x, y)}$ with $V(x, y)$ $\mu$-strongly convex. Let $z = (x, y)$. Strong convexity of $V$ implies that 
\begin{align*}
\left\langle z - z', \nabla V(z) - \nabla V(z') \right\rangle \geq \mu |z - z'|^2
\end{align*}
which is our Assumption~\ref{assmp:convexity} written in a generic notation. We note that this assumption implies (see \cite[Lemma~2]{zhou2018fenchel}) that $V(x, y) - (\mu/2) |y|^2 - (\mu/2) |x|^2$ is convex, which implies that $V(x, y) - (\mu/2) |y|^2$ is convex as it is a sum of two convex functions.

Now consider the log-concave density $h(x,y) \propto e^{-V(x,y) + (\mu/2) |y|^2}$. By Lemma~\ref{lem:PLI}, we can readily see that
\begin{align*}
\tilde{h}(y) = \int h(x, y) \mathrm{d}x \propto e^{(\mu/2)|y|^2}\int e^{-V(x,y)} \mathrm{d}x
\end{align*}
is log-concave. This implies that the function $y \mapsto \int e^{-V(x,y)} \mathrm{d} x$ is $\mu$-strongly log-concave.

\section{Proofs}

\subsection{Proof of Lemma~\ref{lem:strongmonotonicity}}
	With a straightforward calculation Assumption \ref{ass:strongconvexfast} implies the following
	\begin{equation}
	\langle\nabla_x U(\theta,x) - \nabla_x U(\theta,x'), x-x'\rangle \geq \frac{\kappa}{2} |x-x'|^2, \quad \text{for all } x,x' \in \R^{d_x} ,\, \theta \in \R^{d_\theta}.
	\end{equation}
	Hence, taking $x' = 0$,
	\begin{equation}
	\langle\nabla_x U(\theta,x), x \rangle \geq \frac{\kappa}{2} |x|^2 + \langle\nabla_x U(\theta, 0), x\rangle,
	\end{equation}
	so that by multiplying both sides by $-1$ and Assumption \ref{ass:growthcoeffs} we have
	\begin{equation}
	-\langle\nabla_x U(\theta,x), x\rangle \leq - \frac{\kappa}{2} |x|^2 - \langle\nabla_x U(\theta, 0), x\rangle\leq - \frac{\kappa}{2} |x|^2 + C|x| \leq -\frac{\kappa}{4} |x|^2 + \frac{4C^2}{\kappa}
	\end{equation}
	and \eqref{eqn:stronglymonotonicx} is shown with constants independent of $x$ and $\theta$.
	Equation \eqref{eqn:stronglymonotonic2} can be proven with the exact same procedure.

\subsection{Proof of Proposition \ref{thm3.3}}
{We first prove Proposition \ref{thm3.3} \ref{propfirstversion}. This is simply a case of verifying the assumptions of \cite[Theorem 3.3]{crisan2022poisson} hold. Assumption A1 from \cite{crisan2022poisson} holds by \ref{assgenthm1}, \ref{assgenthm2} and \ref{assgenthm3}. Assumption A2 from \cite{crisan2022poisson} holds by \ref{assgenthm3}. Assumption A3 and Assumption A4 from \cite{crisan2022poisson} are satisfied by \ref{assgenthm4}. Finally, \cite[eq (45)]{crisan2022poisson} and \cite[eq (46)]{crisan2022poisson} hold by \ref{assgenthm5}.}

{
		To show Proposition \ref{thm3.3} \ref{propsecondversion}, we must again show that the assumptions of \cite[Theorem 3.3]{crisan2022poisson} hold, with some simplifications due to constant diffusion.
	Assumption A1 from \cite{crisan2022poisson} is implied by \ref{assgenthm1_simple}, \ref{assgenthm2_simple}, and the fact we have constant diffusion. We note that we require less regularity on the drift coefficients here because of the constant diffusion -- this will be addressed below.
		Eq (29) and eq (30) from \cite[Assumption A2]{crisan2022poisson} hold since we have constant diffusion.
		Assumption A3 and Assumption A4 from \cite{crisan2022poisson} are true by \ref{assgenthm4}.
	Eq (46) of \cite[Theorem 3.3]{crisan2022poisson} holds because of \eqref{eqn:2dersimple}, so it is sufficient to show that since we have constant diffusion in both variables, we only need semigroup derivative estimates up to order two for the frozen semigroup (i.e. \cite[eq (59)]{crisan2022poisson} for $k=2$). To see this, we point to  \cite[eq (59)]{crisan2022poisson}. In the constant diffusion case here,
\begin{equation}\label{eqn:derGensimple}\frac{\partial\mathcal{L}^\theta}{\partial \theta_i}\phi^\theta(x) = (\partial_{\theta_i}a(\theta,x), \nabla_x \phi^\theta(x))\end{equation}
	is a first order differential operator, whereas in  \cite[eq (59)]{crisan2022poisson} it is a second order differential operator. This halves the amount of frozen semigroup derivatives eventually needed from four to two meaning that \eqref{eqn:4dersimple} in \ref{assgenthm5_simple} is enough. This also means that instead of 4 orders of regularity on the drift coefficients, we only require 2. Moreover, the extra H\"{o}lder continuity is not required, by the same mollification argument as in \cite[Theorem 7.1]{crisan2022poisson}. Hence, the proof is done.
}

\subsection{Proof of Proposition \ref{prop:averaging}}
	
{We must show that the assumptions of Proposition \ref{thm3.3} \ref{propsecondversion} hold. We note that we have constant diffusion, $b(\theta, x)= -\nabla_{\theta} U(\theta, x)$ and $a(\theta, x)= -\nabla_{x} U(\theta, x)$. Assumptions \ref{assgenthm1_simple} and \ref{assgenthm2_simple} hold by Assumption \ref{ass:growthcoeffs} and since $U \in C^3(\R^{d_\theta} \times \R^{d_x})$ with bounded second and third derivative. Assumption \ref{assgenthm4} holds by Lemma \ref{lem:strongmonotonicity}. Assumption \ref{assgenthm5_simple} holds by Lemma \ref{lemma:fastsemigroupderests} and Lemma \ref{lem:avgderest}.
One might expect the bound \eqref{eqn:uitAveraging} not to be uniform in $\beta$ since the ellipticity of \eqref{eq:theta_update} is not uniform in $\beta$, but the ellipticity of the SP is only needed in \cite{crisan2022poisson} to produce the averaged semigroup derivative estimates (i.e for \cite[Proposition 7.2]{crisan2022poisson}). We obtain these otherwise, see the proof of Lemma \ref{lem:avgderest}, and so do not need this ellipticity to be uniform in $\beta$ to produce a bound independent of $\beta$. Finally, the exponents on the initial conditions in \eqref{eqn:uitAveraging} are lower than those in \cite[Theorem 3.3]{crisan2022poisson} because our drift coefficients have bounded derivative (hence Eq. (123) and Eq. (124) in \cite{crisan2022poisson} are bounded).}

\subsection{Proof of Lemma \ref{lemma:fastsemigroupderests}}
	We prove \eqref{eqn:derestaveraging} and \eqref{eqn:secondDerEstaveraging} using \cite[Lemma 2.8]{angeli2023uniform} and \cite[Lemma 2.10]{angeli2023uniform} respectively. \cite[Lemma 2.10]{angeli2023uniform} would give us the two derivative estimates we require straight away, but because the estimate \eqref{eqn:derestaveraging} has an effect on the constants in \eqref{eqn:slowstrongconvex1}, we must be more careful in ensuring the constants in \eqref{eqn:derestaveraging} are sharp. Hence we use  \cite[Lemma 2.8]{angeli2023uniform} for the proof of \eqref{eqn:derestaveraging}. Assumption 2.6 (17a) from \cite{angeli2023uniform} holds by Assumption \ref{ass:strongconvexfast} above, in particular with $2\gamma = \kappa$. Assumption 2.6 (17b) holds since we have constant diffusion. Hence by \cite[Lemma 2.8]{angeli2023uniform}, \eqref{eqn:derestaveraging} holds.

	Now we show \eqref{eqn:secondDerEstaveraging}, for which we must verify Assumption 2.9 from \cite{angeli2023uniform}. We have already verified (17a) from \cite{angeli2023uniform}. Equation (18a) from \cite{angeli2023uniform} holds by Assumption \ref{ass:growthcoeffs} (\ref{item:bdd3rdDeriv}), and (18b) and (18c) hold because we have constant diffusion. Hence Assumption 2.9 from \cite{angeli2023uniform} is verified and \eqref{eqn:secondDerEstaveraging} holds.

\subsection{Proof of Lemma \ref{lem:avgderest}}
	
	We adapt Lemma 2.10 from \cite{angeli2023uniform} to obtain our required estimates. It is enough to show that the drift coefficient of $\bar \theta^{ \beta}_t$ from \eqref{eq:theta_bar_sde} is strongly convex. That is, there exists $\tilde{\lambda}(\theta)$ such that for all $\xi$,
	\begin{equation}\label{eqn:suffCond1}
	- \sum_{i,j=1}^{d_\theta} \partial_{\theta_i} \left(\int dx \, 
	\partial_{\theta_j} U(\theta,x) \rho_{\theta}(x)\right)\xi_i\xi_j \leq -\tilde{\lambda}(\theta) |\xi|^2,
	\end{equation}
	and this $\tilde{\lambda}(\theta)$ should also satisfy that there exists $\zeta>0$ such that
	\begin{equation}\label{eqn:suffCond2}
	\sum_{i,j,k=1}^{d_\theta} \left|\partial_{\theta_i}\partial_{\theta_j} \left(\int dx \, 
	\partial_{\theta_k} U(\theta,x) \rho_{\theta}(x)\right)\right| \leq \zeta (1+\tilde{\lambda}(\theta)).
	\end{equation}
	These are conditions on the averaged drift, and in particular include the invariant measure of our FrP \eqref{eqn:frozen}. We wish to provide assumptions directly on the function $U$. To this end, we follow the proof of \cite[Proposition 7.2]{crisan2022poisson}, with some simplifications due to our assumptions. To streamline notation, in the following we use $\rho_\theta(f) \coloneqq \int dx \, 
	f(x,\theta) \rho_{\theta}(x)$, and drop the dependencies on $x,\theta$. By \eqref{eqn:derGensimple} and the fact that in our case $a = -\nabla_x U$ we have
    	 \begin{equation}\label{eqn:derGen}\frac{\partial\mathcal{L}^\theta}{\partial \theta_i}\phi^\theta(x) = (\partial_{\theta_i}\nabla_x U(\theta,x), \nabla_x \phi^\theta(x)).\end{equation}
    By using the representation formula (63) in \cite{crisan2022poisson} and \eqref{eqn:derGen} above,
	\begin{align}
	\sum_{i,j=1}^{d_\theta} &\partial_{\theta_i} \left(\rho_{\theta}\left(
	\partial_{\theta_j} U\right)\right)\xi_i\xi_j \\ &=  \sum_{i,j=1}^{d_\theta}  \rho_{\theta}\left(
	\partial_{\theta_i} \partial_{\theta_j} U\right)\xi_i\xi_j 
	+\sum_{i,j=1}^{d_\theta} \int_0^\infty \rho_{\theta}\left(
	(\partial_{\theta_i}\nabla_x U(\theta,x), \nabla_x P_s^{\theta} \partial_{\theta_j}U)\right)\xi_i\xi_j
	\end{align}
	
	By \eqref{eqn:derestaveraging} and Assumption \ref{ass:growthcoeffs}, we have
	\begin{equation}
	\begin{split}
	|(\partial_{\theta_i}\nabla_x U, \nabla_x P_s^{\theta} \partial_{\theta_j}U)| &\leq \left(\sum_{k=1}^{d_x} |\partial_{\theta_i}\partial_{x_k} U |^2 \right)\left(\sum_{k=1}^{d_x} | \partial_{x_k} P_s^{\theta} \partial_{\theta_j}U|^2\right) \\
	&\leq D^2 e^{-\kappa t} .
	\end{split}
	\end{equation}
	Hence we have
	
	\begin{equation}
	\left|\sum_{i,j=1}^{d_\theta} \int_0^\infty ds \rho_{\theta}\left(
	(\partial_{\theta_i}\nabla_x U, \nabla_x P_s^{\theta} \partial_{\theta_j}U)\right)\xi_i\xi_j\right| \leq \frac{d_\theta  D^2}{\kappa}|\xi|^2.
	\end{equation}
	
	Following the calculation in the proof of \cite[Proposition 7.2]{crisan2022poisson} that leads to Equation (174) in that proof, with the simplification in our case \eqref{eqn:derGen}, we obtain \begin{equation}
	\left|\partial_{\theta_i}\partial_{\theta_j} \bar \rho_\theta(\partial_{\theta_k}U)\right| \leq C 
	\end{equation}
	for some $C >0$.

	Eq \eqref{eqn:slowstrongconvex1} above means we can take \begin{equation}\label{eqn:strongconvexConstant}
	\tilde{\lambda}(\theta) = \zeta_0,
	\end{equation}
	so that \eqref{eqn:suffCond1} and \eqref{eqn:suffCond2} hold.
	We have verified (16a) of \cite[Assumption 2.9]{angeli2023uniform}. Since our diffusion is constant we have then verified the whole of \cite[Assumption 2.9]{angeli2023uniform}, and in particular everything is done independently of $\beta$.

\subsection{Statement and proof of concentration bound}\label{subsec:concentration bound}
We now state Proposition \ref{prop:concentration bound} below, which provides a bound on the second term on the RHS of \eqref{triangular inequ}. 

\begin{prop}\label{prop:concentration bound} Let $U$ satisfy Assumption~\ref{assmp:convexity} and $\bar{\theta}_t$ denote the iterate of the process \eqref{eq:averaged_theta_bar}. Then, for every $t$ and every $f \in \mathcal{C}_b^2$, we have
    \begin{equation}\label{eqn:main concentration bouncd}
    \left| \mathbb{E} f(\bar{\theta}_t) - f(\theta^\star)\right|^2 \leq C_f^2\left(e^{-2\mu t} |\bar{\theta}_0 - \theta^\star |^2 +  \frac{2 d_\theta}{\beta} (1 - e^{-2\mu t})\right),
    \end{equation} 
    where $C_f$ denotes the Lipschitz constant of $f$.
\end{prop}
\begin{proof}
    Recall  the diffusion \eqref{eq:averaged_theta_bar}, namely
\begin{align}
    \mathrm{d} \bar{\theta}_t &= \nabla \log p_{\theta_t}(y) \mathrm{d} t + \sqrt{2 \beta^{-1}} \mathrm{d} B_t, \qquad \bar{\theta}_0= \theta_0.\label{eq:averaged_diffusion_1}
\end{align}
To prove our result, observe that $\rho(\theta, x) \propto e^{- U(\theta, x)}$ is $\mu$-strongly log concave as given in Assumption~\ref{assmp:convexity}. By Prekopa-Leindler inequality (See Appendix~\ref{app:PLI}), $\rho(\theta) \propto p_\theta(y) = \int e^{-U(\theta, x)} \mathrm{d}x$ is $\mu$-strongly log concave. Therefore \eqref{eq:averaged_diffusion_1} has a stationary measure $\bar{\rho}_\beta(\theta) \propto e^{\beta \log p_\theta(y)}$, in other words, $\bar{\rho}_\beta(\theta) \propto e^{-\beta G(\theta)}$ with $G(\theta) = - \log p_\theta(y)$

Since $G(\theta)$ is a strongly convex function, we have the relationship
\begin{align*}
\langle \theta - \theta', \nabla G(\theta) - \nabla G(\theta')\rangle \geq \mu |\theta - \theta'|^2,
\end{align*}
and the diffusion can be rewritten as
\begin{align*}
\mathrm{d} \bar{\theta}_t &= - \nabla G(\bar{\theta}_t) \mathrm{d} t + \sqrt{2 \beta^{-1}} \mathrm{d} B_t.
\end{align*}
From  Ito's formula, 
\begin{align*}
\bE |\bar{\theta}_t - \theta^\star |^2 = |\bar{\theta}_0 - \theta^\star|^2 - 2 \mathbb{E}\left[ \int_0^t \langle \bar{\theta}_s - \theta^\star, \nabla G(\bar{\theta}_s)\rangle \mathrm{d} s \right] + {2 d_\theta t}{\beta^{-1}}.
\end{align*}
Hence, using $\nabla G(\theta^\star) = 0$, we have
\begin{align*}
\bE |\bar{\theta}_t - \theta^\star |^2 = |\bar{\theta}_0 - \theta^\star |^2 - 2 \int_0^t \bE[\langle \bar{\theta}_s - \theta^\star, \nabla G(\bar{\theta}_s) - \nabla G(\theta^\star) \rangle] \mathrm{d} s + 2 d_\theta t \beta^{-1},
\end{align*}
Let us take the derivatives of both sides
\begin{align*}
\frac{\mathrm{d}}{\mathrm{d} t} \bE |\bar{\theta}_t - \theta^\star |^2 &= -2 \bE[\langle \bar{\theta}_t - \theta^\star, \nabla G(\bar{\theta}_t) - \nabla G(\theta^\star) \rangle]+ 2d_\theta \beta^{-1} \\
&\leq -2 \mu \mathbb{E} | \bar{\theta}_t - \theta^\star|^2 + 2d_\theta \beta^{-1}
\end{align*}
Given this, we now write
\begin{align*}
\frac{\mathrm{d}}{\mathrm{d} t} e^{2\mu t} \bE |\bar{\theta}_t - \theta^\star|^2 &= 2\mu e^{2\mu t} \bE |\bar{\theta}_t - \theta^\star |^2 + e^{2\mu t} \frac{\mathrm{d}}{\mathrm{d} t} \bE |\bar{\theta}_t - \theta^\star |^2 \\
&\leq 2\mu e^{2\mu t} \bE |\bar{\theta}_t - \theta^\star |^2 - 2\mu e^{2\mu t} \bE |\bar{\theta}_t - \theta^\star |^2 + e^{2\mu t} 2 d_\theta \beta^{-1} \\
&= e^{2\mu t} 2 d_\theta \beta^{-1}.
\end{align*}
Integrating both sides  from $0$ to $t$,
\begin{align*}
e^{2\mu t} \bE |\bar{\theta}_t - \theta^\star|^2 - |\bar{\theta}_0 - \theta^\star|^2 &\leq \frac{2 d_\theta}{\beta} \int_0^t e^{2\mu \tau} \mathrm{d} \tau, \\
&= \frac{2 d_\theta}{\beta} (e^{2\mu t} - 1),
\end{align*}
which implies 
\begin{align}
\bE |\bar{\theta}_t - \theta^\star|^2 \leq e^{-2\mu t} |\bar{\theta}_0 - \theta^\star|^2 +  \frac{2 d_\theta}{\beta} (1 - e^{-2\mu t}).
\end{align}
Now note that any $f\in C_b^2$ will also be Lipschitz, hence
\begin{align*}
\left| \mathbb{E} f(\bar{\theta}_t) - f(\theta^\star)\right| \leq C_f\bE | \bar{\theta}_t - \theta^\star | \leq C_f \left(\bE | \bar{\theta}_t - \theta^\star |^2 \right)^{1/2}.
\end{align*}
\end{proof}

\subsection{Sketch of proof of Proposition~\ref{prop:numerical_error}}\label{sec:discretisation}

 In this section we give a proof of Proposition \ref{prop:numerical_error}.  We recall that    $\zz_k^{\delta}:=(\theta_k^{\delta}, X_k^{\delta})$ is as in Algorithm \ref{alg:ais}, i.e. $\zz_k^{\delta, \epsilon}$ is the Euler-Maruyama discretization of \eqref{eq:theta_update}-\eqref{latent_update}. Here for clarity we denote $\zz_k^{\delta}=(\theta_k^{\delta}, X_k^{\delta})$ by $\zz_k^{\delta, \epsilon}=(\theta_k^{\delta, \epsilon}, X_k^{\delta, \epsilon})$, i.e. we keep track of the dependence on $\epsilon$ on the notation. Similarly, we denote by $\zz^{\epsilon}_t= (\theta_t^{\epsilon}, X_t^{\epsilon})$ the solution of the slow-fast system  \eqref{eq:theta_update}-\eqref{latent_update}.
 
 \begin{proof}
     First of all let us observe that Assumption 2.1 of \cite[Theorem 7.3]{mattingly2002ergodicity} is easily satisfied by ellipticity of the SDE \eqref{eq:theta_update}-\eqref{latent_update} and  \cite[Assumption 2.4]{mattingly2002ergodicity} is satisfied by choosing $U$ as a Lyapunov function (and indeed it is easy to see that $U$ does satisfy the Lyapunov inequality \cite[(2.2)]{mattingly2002ergodicity} if \eqref{ass:euler-estimate} holds). Hence the SDE \eqref{eq:theta_update}-\eqref{latent_update} admits a unique invariant measure, $\pi^{\epsilon}$, and it is  geometrically ergodic as well. More explicitly, using \cite[Theorem 2.5 and calculations on page 193]{mattingly2002ergodicity} one obtains  
$$
\left\vert \mathbb E g(\zz\ep_t) - \pi^{\epsilon}(g)\right \vert \leq C e^{-\lambda t} (1+ U(z_0)). 
$$
In order to prove \eqref{bound-numerical-error} we write
\begin{align}
    \left\vert \mathbb E g(\zz\de_k) - \mathbb E g(\zz_t\ep)\right \vert &\leq \left\vert \mathbb E g(\zz\de_k) - \pi^{\epsilon}(g)\right \vert \label{eulerbound1}\\
    &+ \left\vert \mathbb E g(\zz\ep_t) - \pi^{\epsilon}(g)\right \vert \, . \label{eulerbound2}
\end{align}
  A bound on \eqref{eulerbound2} has already been achieved. Now, for the bound on \eqref{eulerbound1}, we use again \cite[Theorem 7.3]{mattingly2002ergodicity}. \cite[Condition 7.1 and (7.4)]{mattingly2002ergodicity} are satisfied (with $s=1$), by using the calculations on page 193 of that paper. Moreover, by ellipticity, it is standard to show that the minorization condition\cite[(6.6)]{mattingly2002ergodicity} holds for our dynamics (see e.g. \cite[Theorem 6.2]{mattingly2002ergodicity}). Finally, from (7.5) and the bound before (7.4) one concludes the following
\begin{align*}
    \left \vert \mathbb E g(\zz_k\de) - \pi^{\epsilon}(g)\right \vert \leq G U(z_0) e^{-\tilde{\lambda} k \delta} + \tilde{G} \delta^{\xi} \pi^{\epsilon}(U) \,.
\end{align*}
This concludes the proof.
 \end{proof}


\section{Ascent property of the EM algorithm}\label{app:EM}
One can show that the EM algorithm maximises the likelihood iteratively, by noting that given a parameter estimate $\theta_k$, we have
\begin{align*}
\log p_\theta(y) &= \log \int p_{\theta}(x, y) \md x\\
&= \log \int \frac{p_{\theta_k}(x, y)}{p_{\theta_k}(x,y)} p_\theta(x, y) \md x \\
&= \log p_{\theta_k}(y) + \log \int {p_{\theta_k}(x|y)}\frac{p_\theta(x,y)}{p_{\theta_k}(x, y)} \md x \\
&\geq \log p_{\theta_k}(y) + \underbrace{\int p_{\theta_k}(x|y) \log {p_\theta(x,y)} \md x -  \int p_{\theta_k}(x|y) \log {p_{\theta_k}(x,y)} \md x}_{\Delta(\theta, \theta_k)}
\end{align*}
where the third equality in the above comes from writing $p_{\theta_k}(x,y)=p_{\theta_k}(x\vert y)  p_{\theta_k}(y)$, and 
\begin{align*}
\Delta(\theta, \theta_k) := Q(\theta, \theta_k) - Q(\theta_k, \theta_k).
\end{align*}
Now let $\theta_{k+1} \in \argmax_\theta Q(\theta, \theta_k)$. Then
\begin{align*}
&\log p_{\theta_{k+1}}(y) \geq \log p_{\theta_k}(y) + \underbrace{Q(\theta_{k+1}, \theta_k) - Q(\theta_k, \theta_k)}_{\geq 0}.
\end{align*}
Therefore, ascent in the marginal likelihood is achieved by this method.

\end{document}